\documentclass[runningheads]{llncs} % or whatever
\usepackage[pass]{geometry}
\usepackage[bookmarks=true]{hyperref}
\usepackage{amssymb}
\usepackage{etex} % solve "No room for a new dimen"
\usepackage{caption}
\usepackage{subcaption}

\usepackage[cmex10]{amsmath}

\usepackage{ntheorem}
\theoremstyle{break}

\usepackage{graphicx}
\usepackage{verbatim}
\usepackage{tikz}
\usetikzlibrary{arrows,positioning,shadows}
\usepackage{pgf,pgfarrows,pgfnodes,pgfautomata,pgfheaps,pgfplots}
\usepackage{pst-plot}   %coordinate system

\usepackage[ruled,vlined,linesnumbered]{algorithm2e}

\usepackage{authblk}

\usepackage{wrapfig}

\usepackage{url}
\urldef{\mailsa}\path|{alfred.hofmann, ursula.barth, ingrid.haas, frank.holzwarth,|
\urldef{\mailsb}\path|anna.kramer, leonie.kunz, christine.reiss, nicole.sator,|
\urldef{\mailsc}\path|erika.siebert-cole, peter.strasser, lncs}@springer.com|

\newcommand{\vol}{\operatorname{Vol}}

\begin{document}

\mainmatter  % start of an individual contribution

\title{Analytic Theory to Differential Privacy}

\author{Genqiang Wu\inst{1,2} \and Xianyao Xia\inst{2} \and Yeping He\inst{2}}

\institute{SIE, Lanzhou University of Finance and Economics, Lanzhou 730020, China 
\and
Institute of Software Chinese Academy of Sciences, Beijing 100190, China \\
\email{genqiang80@gmail.com, \{xianyao,yeping\}@nfs.iscas.ac.cn}
}

\authorrunning{Genqiang Wu, Xianyao Xia and Yeping He}

\maketitle

\begin{abstract}
The purpose of this paper is to develop a mathematical analysis theory to solve differential privacy problems. 
The heart of our approaches is to use analytic tools to characterize the correlations among the outputs of different datasets, which makes it feasible to represent a differentially private mechanism with minimal number of parameters. These results are then used to  construct differentially private mechanisms analytically. 
Furthermore, our approaches are universal to almost all query functions.
We believe that the approaches and results of this paper are indispensable complements to the current studies of differential privacy that are ruled by the ad hoc and algorithmic approaches.
\end{abstract}

\begin{keywords}
differential privacy, analytic theory, universal mechanism, optimal mechanism %utility-privacy tradeoff
\end{keywords}
%\newpage

\section{Introduction}  \label{section-introduction}

Differential privacy \cite{DBLP:conf/tcc/DworkMNS06,DBLP:conf/icalp/Dwork06} studies how to query dataset while preserving the privacy of individuals whose sensitive information is contained in the dataset. The crux of differential privacy is to find efficient algorithms, called (privacy) mechanisms, to query sensitive dataset to obtain relatively accurate outputs while satisfying differential privacy.

Since the introduction of differential privacy in 2006 \cite{DBLP:conf/tcc/DworkMNS06,DBLP:conf/icalp/Dwork06}, it has obtained intensive attentions, both from academic community \cite{DBLP:journals/fttcs/DworkR14,salil-vadhan-cdp,DBLP:journals/spm/SarwateC13,DBLP:journals/tkde/ZhuLZY17} and from industry world\footnote{\url{http://devstreaming.apple.com/videos/wwdc/2016/709tvxadw201avg5v7n/709/709_engineering_privacy_for_your_users.pdf}}  \cite{DBLP:conf/ccs/ErlingssonPK14}.
However, some fundamental problems in differential privacy are still not solved: 
``Current approaches are \emph{too ad hoc} and frequently
do \emph{not achieve anything close to the best tradeoffs} between accuracy and privacy'', as commented by Dwork et al.\cite{DBLP:journals/corr/AbowdADKMR17}.
Specifically, first, there are too many different problems in differential privacy and each problem has very different mechanisms from others'. 
For example, there are the data publishing problems \cite{DBLP:conf/sigmod/ZhangXX16,DBLP:conf/kdd/MohammedCFY11,DBLP:conf/ccs/ChenAC12,DBLP:journals/pvldb/ChenMFDX11}, the principal component analysis problems \cite{DBLP:conf/stoc/HardtR12,DBLP:conf/stoc/DworkTT014,DBLP:conf/soda/KapralovT13}, the classification problems \cite{DBLP:journals/jmlr/ChaudhuriMS11,DBLP:journals/pvldb/ZhangZXYW12} and the clustering problems \cite{DBLP:conf/codaspy/SuCLBJ16} among others in differential privacy, and 
it would be hard to imagine that the mechanisms of data publishing can be used to solve other problems mentioned above. 
Second, even though there are some mechanisms, such as the sensitivity-based mechanisms \cite{DBLP:conf/tcc/DworkMNS06,DBLP:conf/stoc/NissimRS07,DBLP:conf/sigmod/ZhangCPSX15} and the exponential mechanism \cite{DBLP:conf/focs/McSherryT07}, are universal to deal with many differential privacy problems, these mechanisms are far less optimal to the tradeoffs between  utility and privacy.

Then, the questions are: whether there exist approaches that are both universally applicable and (near) optimal to most differential privacy problems? if there do exist, how can we construct them? These are challenging questions since the current works in differential privacy don't give any obvious hint about how to deal with them. To answer these questions, there are many works to be done.

\subsection{Outline and Contribution}

This paper mainly do three works to answer the above questions. 

First, since the current works in differential privacy are done in a separated and informal way, a formal and universal illustration to differential privacy problems is needed to pave the way to study differential privacy rigorously and universally.
We will give an abstract model to mathematically formalize differential privacy problems in a universal way, which will show that a query function in differential privacy, in fact, is an \emph{operator} in functional analysis \cite{Erwin-Kreyszig1978}. These materials will be presented in Section \ref{section-DP-model}.

Second, we will study the common properties or structures of differential privacy mechanisms that are applicable to most differential privacy problems. 
The complexity of constructing a differential privacy mechanism $\mathcal M$ is mainly due to that 
 the outputs $\mathcal M(x), x\in\mathcal D$, where each $\mathcal M(x)$ is a probability distribution and $\mathcal D$ is the universe of datasets, are strongly correlated to each other, which defeats most attempts to change some $\mathcal M(x)$'s to improve utility while letting the results still satisfy differential privacy. We will introduce analytic tools to characterize these correlations and then simplify the approaches to construct mechanisms. Especially, we will discuss what are the minimal amount of parameters that are needed to represent a differential privacy mechanism. Our approaches, in principle, are motivated by some approaches in functional analysis \cite{Erwin-Kreyszig1978}.
These problems are discussed in Section \ref{section-representation-theory}.

Third, there should be approaches to balance privacy and utility. 
These approaches should both be applicable for  and (in some extent) optimal to most differential privacy problems in constructing mechanisms.
By using the results in Section \ref{section-representation-theory} we can make changes to some or all $\mathcal M(x), x\in\mathcal D$ while letting the results still satisfy differential privacy, which makes it feasible to construct mechanisms to balance privacy and utility at will.
We will discuss these problems in Section \ref{section-mechanism-design}.

In brief, we will develop a mathematical analysis theory to study differential privacy.

\section{Preliminaries}  \label{section-prelimiaries}

\subsection{Notations}

The notational conventions of this paper are summarized in Table \ref{table-1}.
In this paper, unless noted otherwise, any set is \emph{not} a multiset. Set $\epsilon >0$.

\begin{table}[t]
\caption{Table of Notation}  \label{table-1}
\begin{center}
    \begin{tabular}{ | l | l |}
    \hline
    Notation & Description  \\ \hline \hline
    $\|z\|_p$ & the $\ell_p$-norm of the real vector $z$ \\ \hline
    $\bar{\mathbb N}$       & the set of the natural numbers including 0 \\  \hline
    $\mathcal D, \bar d$   &  the set of possible datasets, the metric on the set $\mathcal D$, respectively   \\ \hline
    $\mathcal R, d$  &  the set $\{f(x):x\in \mathcal D\}$, the metric on the set $\mathcal R$, respectively \\ \hline     
    $\mathcal P(\mathcal R)$ & the set of possible probability distributions on $\mathcal R$ \\  \hline
    $\mathcal M(x)$ & a probability distribution in $\mathcal P(\mathcal R)$ or a random variable following it \\ \hline
    %the output probability distribution of the mechanism $\mathcal M$ when the input dataset is $x$ \\
    $p^x(r)$    & the probability distribution of the random variable $\mathcal M(x)$ \\ \hline
    $q^x(r)$    &  the discretized probability distribution of $p^x(r)$ \\ \hline
    $\mathcal R_i^x$ & the $i$th layer of the probability distribution $p^x(r)$ or $q^x(r)$ \\ \hline
    $\mathcal N_i^x$ & the set $\{ y \in \mathcal D: i-1 < \bar d(x,y) \le i \}$ \\   \hline
    ${L}_{q}^{x}(r)$    &    the number of layer of $r$ in the probability distribution $q^x(r)$ \\ \hline
    $A-B$  &   the set difference $\{x: x\in A \wedge x\notin B\}$  \\ \hline
    $A-y$   &   the set $\{x-y:x\in A\}$ \\ \hline
    $N_{\delta}^{f(x)}$   &   the $\delta$-neighborhood of $f(x)$, i.e., the set $ \{r\in \mathcal R: d(r,f(x)) \le \delta\}$\\ \hline
    $\mathbb B$  & the universe of $\epsilon$-differential privacy mechanisms of the function $f$ \\ \hline
    $\bar{\mathbb B}$ &  the universe of the set sequences of mechanisms in $\mathbb B$ \\ \hline
    $\mathbb C$  & the universe of set sequences of $f$ satisfying (\ref{equation-7})\\ \hline
    \end{tabular}
\end{center}
\end{table}

\subsection{Differential Privacy}  \label{section-preliminary-dp-model}

A dataset is a collection (a multiset) of $n$ records, each of which is derived from the record universe $\mathcal X$ and denotes the information of one individual. We use the histogram representation $x \in \mathbb {\bar N}^{|\mathcal X|}$ to denote the dataset $x$, where the $i$th entry $x_i$ of $x$ represents the number of elements in $x$ of type $i \in \mathcal X$ \cite{DBLP:conf/stoc/NikolovTZ13,DBLP:journals/fttcs/DworkR14,DBLP:journals/vldb/LiMHMR15}. Two datasets $x,x'$ are said to be \emph{neighbors (or neighboring datasets)} if $\|x-x'\|_1=1$.

Differential privacy \cite{DBLP:conf/icalp/Dwork06,DBLP:journals/fttcs/DworkR14} characterizes privacy by capturing the changes of outputs when one's record in the queried dataset is changed. The changes of datasets are modeled by the notion of the neighboring datasets. For the dataset universe $\mathcal D$ and a query function $f$, let $\mathcal R\supseteq\{f(x): x \in \mathcal D\}$ and equip $\mathcal R$ with a Borel $\sigma$-algebra $\mathcal B$ \cite{athreya2006measure}.

\begin{definition}[$\epsilon$-Differential Privacy]
For the dataset universe $\mathcal D$, let $\mathcal P(\mathcal R)$ denote the set of all the probability measures on $(\mathcal R,\mathcal B)$. A mapping $\mathcal M: \mathcal D \rightarrow \mathcal P(\mathcal R)$ gives \emph{$\epsilon$-differential privacy} if for any two neighbors $x, x' \in \mathcal D$, and any $S \in \mathcal B$, there is
\begin{align} \label{formula-inequality-dp}
\Pr[ \mathcal M(x) \in S] \le \exp(\epsilon) \Pr[ \mathcal M(x')\in S],
\end{align}
where we abuse the notation $\mathcal M(x)$ as either denoting a probability distribution in $\mathcal P(\mathcal R)$ or denoting a random variable following the probability distribution.

\end{definition}

\subsection{Achieving Differential Privacy}

The global sensitivity-based method is a basic approach to achieve differential privacy \cite{DBLP:conf/tcc/DworkMNS06}. We first define the global sensitivity.

\begin{definition}[Global Sensitivity]
For the query function $f: \mathcal D\rightarrow \mathcal R$, if $\mathcal R\subseteq \mathbb R^k$, then the global sensitivity of $f$ is defined as
\begin{align*}
\Delta f = \max_{x,x'\in \mathcal D:\|x-x'\|_1= 1}\|f(x)-f(x')\|_1.
\end{align*}
Furthermore, letting $s^x: \mathcal R \rightarrow \mathbb (-\infty,0]$ be the \emph{score function} when the inputted dataset is  $x$, the global sensitivity of $s^x$ is defined as
\begin{align*}
\Delta s= \max_{r\in \mathcal R,x,x'\in \mathcal D:\|x-x'\|_1= 1} |s^x(r)-s^{x'}(r)|.
\end{align*}
\end{definition}

The Laplace mechanism \cite{DBLP:conf/icalp/Dwork06,DBLP:journals/fttcs/DworkR14} is one important global sensitivity-based mechanism and the Exponential mechanism \cite{DBLP:conf/focs/McSherryT07} is one generalization of the global sensitivity-based mechanisms.

\begin{definition} \label{definition-laplace-gaussian-exponential}
The \emph{Laplace mechanism} $\mathcal{M}(x)$ generates a real random vector $r =(r_1, \ldots, r_k)$ with probability distribution
\[ p^x(r) = \frac{\epsilon}{2\Delta f} \exp \left(-\frac{\epsilon\|r-f(x)\|_1}{\Delta f}\right). \]
The \emph{Exponential mechanism} $\mathcal{M}(x)$ outputs an element $r \in \mathcal R$ with probability distribution \begin{align}
p^x(r)=\frac{1}{\alpha^x}\exp\left(\frac{\epsilon s^x(r)}{2\Delta s}\right),
\end{align}
where $\alpha^x$ is the normalizor. 
\end{definition}

Both the Laplace mechanism and the Exponential mechanism satisfy $\epsilon$-differential privacy \cite{DBLP:journals/fttcs/DworkR14}.

\subsection{A Lemma to Mediant Inequalities}

%The following lemma will be used frequently in this paper.

\begin{lemma}  \label{chap-optimal-mechanism:lemma-1}
Let $g(x)=\frac{\alpha_0+\alpha_1x}{\beta_0+\beta_1x}, x\in \mathbb R$, where $\alpha_i \ge 0, \beta_i >0$ for $i\in\{0,1\}$.  If $\frac{\alpha_0}{\beta_0} <\frac{ \alpha_1}{\beta_1}$, then $g(x)$ is increasing. Otherwise, if $\frac{\alpha_0}{\beta_0} \ge\frac{ \alpha_1}{\beta_1}$, then $g(x)$ is decreasing.
\end{lemma}
\begin{proof}
Note that the derivative of $g(x)$ is $g'(x)=\frac{\alpha_1\beta_0-\alpha_0\beta_1}{(\beta_0x+\beta_1)^2}$, by which the claims are immediate.
\qed
\end{proof}

\section{Abstract Model of Differential Privacy}  \label{section-DP-model}

In this section we present an abstract model to differential privacy, of which the intentions are to formalize differential privacy problems in a universal and formal way, and to pave the way to discuss the common properties of differential privacy problems.
There are somewhat similar treatments in \cite{DBLP:journals/isci/HolohanLM15,DBLP:conf/pet/ChatzikokolakisABP13}. 

We will model each differential privacy problem as a problem of a query function $f: \mathcal D \rightarrow \mathcal R$, of which the domain $\mathcal D$ is the set of possible datasets  on which a metric is defined, and of which 
$\mathcal R \supseteq \{f(x):x\in \mathcal D\}$ on which another one metric is defined. That is, a query function in differential privacy, in general, is an \emph{operator} in functional analysis \cite{Erwin-Kreyszig1978}. For the simplicity of presentation, in the following parts of this paper, we set $\mathcal R =\{f(x):x\in \mathcal D\}$.

\subsection{Dataset Metric Space and Value Metric Space}  \label{subsection-theory-metric-space}

The dataset universe is modeled as a set $\mathcal D$ on which a metric $\bar d$ is defined.\footnote{The definitions of the metric and the metric space follow the references \cite{Micheal2007metric-space,Erwin-Kreyszig1978}.}

\begin{definition}[Dataset Metric Space] \label{definition-dataset-metric-space}
Let $f$ be a function defined on the set $\mathcal D$ on which a metric $\bar d$ is defined. Then the metric space $(\mathcal D, \bar d)$ is called the \emph{dataset metric space} of $f$. Two elements $x,y \in \mathcal D$ are said to be neighbors (or neighboring datasets) of distance $k$ if $k-1 < \bar d(x,y) \le k$, for $k \in \bar{ \mathbb N}$. When $k= 1$, $x,y$ are said to be neighbors (or neighboring datasets).
\end{definition}

Set $\mathcal N_i^x=\{ y \in \mathcal D: i-1 < \bar d(x,y) \le i \}$, for $i \in \bar{\mathbb N}$. Set $\mathcal {\bar N}_i^{x}=\{y \in \mathcal D:\bar d(x,y)\le i\}$ for $i \in \bar{\mathbb N}$ and set $\mathcal N^x=\{ y \in \mathcal D: \bar d(x,y) \le 1 \}$ for abbreviation. 

The codomain of the query function $f$ on $\mathcal D$ is modeled as a set $\mathcal R$ on which a metric $d$ is defined.
\begin{definition}[Value Metric Space] \label{definition-value-metric-space}
For a function $f$ on $\mathcal D$, set $\mathcal R=\{f(x): x\in \mathcal D\}$. Defining a metric $d$ on $\mathcal R$, then $(\mathcal R, d)$ is called the \emph{value metric space} of $f$. Equipping $\mathcal R$ with the Borel $\sigma$-algebra $\mathcal B$ generated by the open sets in $\mathcal R$ (in the metric topology), then $(\mathcal R, \mathcal B)$ is a measurable space \cite{athreya2006measure}.
\end{definition}

The product metric space and the product probability space are used to model the batch query functions.

\begin{definition}[Product Metric Space]
If $(\mathcal R_1,d_1),\ldots,(\mathcal R_n,d_n)$ are metric spaces, and $N$ is a norm on $\mathbb R^n$, then $\big(\mathcal R_1\times \cdots \times \mathcal R_n, N(d_1,\ldots,d_n)\big)$ is a metric space, where the \emph{product metric} $N(d_1,...,d_n)$ is defined by
\[N(d_1,...,d_n)\big((x_1,\ldots,x_n),(y_1,\ldots,y_n)\big) = N\big(d_1(x_1,y_1),\ldots,d_n(x_n,y_n)\big),\]
and the induced topology agrees with the \emph{product topology}.
\end{definition}

\begin{definition}[Product Probability Space]
Let $(\mathcal R_1, \mathcal B_1, \mu_1),$ $\ldots,$  $(\mathcal R_n, \mathcal B_n, \mu_n)$ be $n$ probability spaces. Then the probability space $(\mathcal R, \mathcal B, \mu)$, defined by $\mathcal R= \mathcal R_1 \times \cdots \mathcal \times \mathcal R_n$, $\mathcal B = \mathcal B_1 \times \cdots \times \mathcal B_n$ and $\mu= \mu_1 \times \cdots \times \mu_n$, is called the \emph{product probability space} of the $n$ probability spaces.
\end{definition}

For $n$ query functions $f_1,\ldots, f_n$ over the dataset metric space $(\mathcal D,\bar d)$, let $(\mathcal R_1,d_1),$ $\ldots, (\mathcal R_n,d_n)$ be their value metric spaces, respectively. Then the product metric space $\big(\mathcal R_1\times \cdots \times \mathcal R_n, N(d_1,\ldots,d_n)\big)$ is called the (product) value metric space of $f_1,\ldots, f_n$.

\subsection{Definition of Differential Privacy}

For the query function $f:\mathcal D \rightarrow \mathcal R$, let $(\mathcal D, \bar d)$, $(\mathcal R, d)$ be its dataset metric space and value metric space, respectively. Let $(\mathcal R, \mathcal B)$ be a measurable space. 

\begin{definition}[$\epsilon$-Differential Privacy]  \label{definition-abstract-dp}
%For a dataset metric space $(\mathcal D, \bar d)$, 
Let $\mathcal P(\mathcal R)$ denote the set of all the probability distributions on $(\mathcal R,\mathcal B)$. A mapping $\mathcal M: \mathcal D \rightarrow \mathcal P(\mathcal R)$ gives \emph{$\epsilon$-differential privacy} if, for any two neighbors $x, x' \in \mathcal D$ and any $S \in \mathcal B$, there is
\begin{align} \label{formula-inequality-dp}
\Pr[ \mathcal M(x) \in S] \le \exp(\epsilon) \Pr[ \mathcal M(x')\in S].
\end{align}
%where we abuse the notation $\mathcal M(x)$ as either denoting a probability distribution in $\mathcal P(\mathcal R)$ or denoting a random variable following the probability distribution. %\alert{The mapping $\mathcal M$ is said to be a mechanism to achieve DP.}
\end{definition}

For the random variable $\mathcal M(x)$, let $p^x(r)$ be its probability distribution. Then the mechanism $\mathcal M$ can be represented by the set $\left\{p^x(r): x\in \mathcal D \right\}$.

\begin{proposition}[Composition Privacy]  \label{proposition-composition-privacy}
For the dataset metric space $(\mathcal D, \bar d)$, let $\mathcal M_i$ be $\epsilon_i$-differentially private on $(\mathcal R_i,\mathcal B_i)$ for $i \in \{1,\ldots,n\}$. Then the composition of $\mathcal M_1,\ldots,\mathcal M_n$, defined by $\mathcal M(x)= (\mathcal M_1(x),\ldots, \mathcal M_n(x))$, $x \in \mathcal D$, satisfies $\sum_{i=1}^n \epsilon_i$-differential privacy on the product probability space $(\mathcal R,\mathcal B)$.
\end{proposition}

\begin{proof}
The proof is similar with the one of Theorem 3.14 in \cite{DBLP:journals/fttcs/DworkR14} and is omitted. %\alert{The proof seems needing advanced techniques for the dependent cases.}
\qed
\end{proof}

\begin{proposition}[Group Privacy]
Let $\mathcal M$ be an $\epsilon$-differentially private mechanism. Assume that, for any $x,y \in \mathcal D$ with $\bar d(x,y)=i$ for $i \in \mathbb N$, there exists $x' \in \mathcal D$ such that $\bar d(x,x')=1$ and $\bar d(x',y)=i-1$. Then
\[ \Pr[\mathcal M(x) \in S] \le \exp(i\epsilon) \Pr[\mathcal M(y) \in S], \]
for any $S \in \mathcal B$.
\end{proposition}

\begin{proof}
The proposition is an immediate corollary of the inequality (\ref{formula-inequality-dp}).
\qed
\end{proof}

\subsection{Utility and Optimal Mechanism Problem}  \label{subsection-theory-utility}

%In the section we see how to measure the accuracy or noise complexity of mechanisms.

Let $(\mathcal R,d)$ be the value metric space of $f$.  %Set $C_T^x= \{r \in \mathcal R :d(f(x),r) \le T\}$ for all $x \in \mathcal D$, where $T>0$.
We use the expected distortion $P^x$ between the random variable $\mathcal M(x)$ and $f(x)$  to measure the utility of the mechanism $\mathcal M$ at the dataset $x$, i.e.,
\begin{align}  \label{equation-16}
P^x= \mathbb E_p[d(\mathcal M(x),f(x))] = \int_{r \in \mathcal R} d(r,f(x)) p^x(dr),
\end{align}
where $p^x(r)$ is the probability distribution of $\mathcal M(x)$.
We use the set $\{ P^x: x\in \mathcal D\}$ to measure the utility of $\mathcal M$. An alternative to measure the utility of $\mathcal M$  is to use the expected value of $P^x$, i.e.,
\begin{align}  \label{equation-19}
P=\mathbb E P^x=\int_{x\in \mathcal D} P^xp(dx)=\int_{x\in \mathcal D} \int_{r \in \mathcal R} d(r,f(x)) p^x(dr)p(dx),
\end{align}
where $p(x)$ is the occurring probability distribution of datasts in $\mathcal D$. 

Let the set 
\begin{align} \label{equation-26}
\mathbb B= \left\{\{p^x(r):  x\in \mathcal D\}: \{p^x(r):  x\in \mathcal D\} \mbox{ satisfies $\epsilon$-differential privacy}\right\}
\end{align}
denote the universe of $\epsilon$-differential privacy mechanisms of the query function $f$. Then, the \emph{Pareto optimal mechanism problem} of the query function $f$ would be the multi-objective optimization problem \cite{citeulike:3789084}
\begin{align} \label{equation-23}
\min_{\{p^x(r):  x\in \mathcal D\} \in \mathbb B}\{ P^x: x \in \mathcal D\}.
\end{align}
Similarly, the \emph{expected optimal mechanism problem} of the query function $f$ would be the optimization problem
\begin{align} \label{equation-25}
\min_{\{p^x(r):  x\in \mathcal D\} \in \mathbb B}P.
\end{align}

\subsection{Query Function}  \label{subsection-theory-function}

Notice that the definition of a query function in above sections is consistent with the definition of an \emph{operator} in functional analysis \cite{Erwin-Kreyszig1978}. Therefore, following the tradition of functional analysis, a query function $f:\mathcal D\rightarrow \mathcal R$ is also called an \emph{operator} $f:\mathcal D\rightarrow \mathcal R$.

The linear function is known to be one kind of the simplest query functions in differential privacy, which is a generalization of the sum function or the counting function \cite{DBLP:journals/fttcs/DworkR14}.

\begin{definition}[Linear Function]  \label{definition-linear-function}
For the query function $f:\mathcal D \rightarrow \mathcal R$, assume both of $\mathcal D, \mathcal R$ are vector spaces over the same field \cite{Erwin-Kreyszig1978}. If
$f(x+y)=f(x)+f(y)$ for all $x,y \in \mathcal D$, the function $f$ is said to be a linear (query) function or a linear operator \cite{Erwin-Kreyszig1978}.
\end{definition}

Note that, for a linear function $f$, the set
\begin{align*}
\mathcal V_f :=& \{f(x')-f(x): x' \in \mathcal N^x\} \\
=& \{f(x'-x): x' \in \mathcal N^x\}  \\
=&  \{f(y): \bar d(0,y) \le 1\}
\end{align*}
has no difference for different $x$ and seems to be different from those of other linear queries.
Therefore, we can use $\mathcal V_f$ to represent $f$. We call $\mathcal V_f$ \emph{the neighboring set} of the linear function $f$.
Any query function, which is not a linear function, is said to be a non-linear (query) function.

\begin{definition}[Monotonic Function]  \label{definition-monotonic-function}
The function $f$ is said to be a \emph{monotonic (query) function} if for any $x \in \mathcal D$ and all $y,z \in \mathcal D$ such that $\bar d(x,y)> \bar d(x,z)$, there is $d(f(x),f(y))\ge d(f(x),f(z))$, and to be a \emph{strictly monotonic (query) function} if $d(f(x),f(y))> d(f(x),f(z))$.
\end{definition}

The identity function is used to model the data publication problem in differential privacy \cite{DBLP:conf/ccs/ChenAC12,DBLP:journals/pvldb/ChenMFDX11}.

\begin{definition}[Permutation Function and Identity Function] \label{definition-permutation-function}
The injective function $f:\mathcal D \rightarrow \mathcal R$ is called a permutation function if $\mathcal D=\mathcal R$.
%if its value metric space $(\mathcal R, d)$ is the same as its dataset metric space $(\mathcal D,\bar d)$. 
Moreover, if $f(x)=x$ for all $x \in \mathcal D$, then $f$ is called an identity function.
\end{definition}

%\alert{Delete?}

\begin{definition}[Global Sensitivity and Local Sensitivity]
Let $(\mathcal D,\bar d)$, $(\mathcal R,d)$ be the dataset metric space, the value metric space of the function $f$, respectively. The global sensitivity of $f$ is defined as
\[  \Delta f = \max_{x,x'\in \mathcal D:\bar d(x',x) \le 1} d(f(x),f(x')). \]
The local sensitivity of $f$ at $x$ is defined as
\begin{align*}
\Delta^x_{f}=\max_{x' \in \mathcal D:\bar d(x',x) \le 1} d(f(x),f(x')).
\end{align*}
%The local sensitivity of distance $i$ of $f$ at $x$ is defined as
%\begin{align*}
%\Delta^x_{i} = \Delta^x_{i,f}=\max_{x' \in \mathcal D:\bar d(x',x) \le i} \Delta^{x'}_{f}.
%\end{align*}
%Obviously, $\Delta_f^x = \Delta_0^x $.
\end{definition}

Note that the definitions of the global sensitivity and the local sensitivity are consistent with those in  \cite{DBLP:conf/tcc/DworkMNS06,DBLP:conf/stoc/NissimRS07,DBLP:conf/sigmod/ZhangCPSX15}.

\subsection{Instance Interpretation}  \label{subsection:instance-interpretation}

The abstract model in Section \ref{section-DP-model} is consistent with the classic differential privacy model. To see this, we give some instance interpretations. 

First, for a query function, the set $\mathcal D$ in the dataset metric space $(\mathcal D, \bar d)$ is equivalent to the dataset universe in the classic differential privacy model. The metric $\bar d$ captures the mathematical meaning of the neighboring relation of datasets. The details are as follows. There are two different definitions about neighboring datasets in differential privacy: \emph{bounded neighboring datasets} and \emph{unbounded neighboring datasets} \cite{DBLP:conf/sigmod/KiferM11}. For the definition of bounded neighboring datasets, all of the datasets are assumed to have the same number $n$ of records. Two datasets $x,x' \in \mathcal D$ are said to be neighboring datasets if $\|x-x'\|_1 =2$, where $x,x'$ are their histogram representations as noted in Section \ref{section-preliminary-dp-model}. In this case, we can set $\bar d(x,y)= \frac{\|x-y\|_1}{2}$. For the definition of unbounded neighboring datasets, the number of records in a dataset can be any natural number. Two datasets $x,x' \in \mathcal D$ are said to be neighboring datasets if $\|x-x'\|_1 =1$. In this case, we can set $\bar d(x,y)= \|x-y\|_1$.

Theoretically, in differential privacy, almost all of data processing problems can be explained as a function $f$ whose domain is set to be $\mathcal D$ and whose codomain is set to be $\mathcal R$, such as the SQL query problems \cite{DBLP:conf/icdm/HayLMJ09,DBLP:conf/tcc/KasiviswanathanNRS13,DBLP:conf/asiacrypt/DworkNRR15,DBLP:conf/stoc/HardtT10}, the statistical problems \cite{DBLP:conf/stoc/DworkL09,DBLP:journals/corr/DworkS015,DBLP:conf/focs/RogersRST16}, and the data ming or machine learning problems
\cite{DBLP:conf/kdd/McSherryM09,DBLP:conf/kdd/MohammedCFY11,DBLP:conf/nips/ChaudhuriHS14,DBLP:journals/jmlr/ChaudhuriSS13,DBLP:conf/nips/ChaudhuriV13,DBLP:journals/jmlr/ChaudhuriMS11,DBLP:conf/stoc/DworkTT014,DBLP:journals/pvldb/ZhangZXYW12,DBLP:conf/sigmod/ZhangCPSX14}.
The idea of differential privacy to preserve privacy can be explained as follows: When the real dataset is $x$, in order to preserve privacy, a differentially private mechanism first samples a dataset $y \in \mathcal D$ (according to a probability distribution) and then outputs $f(y)$ as the final query result of $f$. There should be a distortion function $D^x(r)$ to measure the distortion when querying the dataset $x$ but obtaining $f(y)$. Then, the metric $d$ of the value metric space $(\mathcal R, d)$ can be set as
\begin{align}
d(f(x),r)= D^x(r), \hspace{1cm} x\in \mathcal D, r\in \mathcal R,
\end{align}
which is a measure of the distortion of querying the dataset $x$ but obtaining $r$. For example, if $\mathcal R\subseteq \mathbb R^k$, we can set $d(f(x),r)=\|f(x)-r\|_p$\cite{DBLP:conf/stoc/NikolovTZ13,DBLP:journals/vldb/LiMHMR15}; if $\mathcal R$ is a set of real matrices, we can use a norm over matrices to define $d$ \cite{DBLP:conf/stoc/HardtR12,DBLP:conf/stoc/DworkTT014,DBLP:conf/soda/KapralovT13}. For the case of $\mathcal R$ being a set of non-numeric elements, the corresponding distortion function $D^x(r)$, in general, will not satisfy the triangular inequality property and the symmetric property of metric \cite{Micheal2007metric-space}, such as those in \cite{DBLP:conf/kdd/MohammedCFY11,DBLP:conf/ccs/ChenAC12,DBLP:journals/pvldb/ChenMFDX11}. In this condition, the metric $d$ can be considered as an approximation of $D^x(r)$ and we would treat the real problem by the method found when treating the ideal problem, which would simplify the complexity of complex problems. We now give some examples.

\begin{example}[counting query]
A counting query function $f$ outputs non-negative integers. Then, we can set $\mathcal R=\{0,\ldots, k\}$, and $d(r,f(x))=|r-f(x)|$. The subgraph counting query  \cite{DBLP:journals/tods/KarwaRSY14,DBLP:conf/tcc/KasiviswanathanNRS13,DBLP:conf/sigmod/ChenZ13,DBLP:conf/icdm/HayLMJ09} is a special kind of counting queries.
\end{example}

\begin{example}[multi-linear queries \cite{DBLP:conf/stoc/NikolovTZ13,DBLP:journals/fttcs/DworkR14,DBLP:journals/vldb/LiMHMR15,DBLP:conf/nips/WangFZW13}]
For $k$ real valued linear queries $f_1, \ldots, f_k$, we can set $f(x)=(f_1(x),\ldots,f_k(x))$ for all $x\in \mathcal D$. Then $\mathcal R$ can be set as a subset of $\mathbb R^k$ and $d(r,f(x))=\|r-f(x)\|_{p}$.
\end{example}

\begin{example}[data publishing  and synthetic dataset generation]
For the data publishing problem \cite{DBLP:conf/kdd/MohammedCFY11,DBLP:conf/ccs/ChenAC12,DBLP:journals/pvldb/ChenMFDX11,DBLP:conf/kdd/ChenFDS12,DBLP:conf/sigmod/ZhangXX16,DBLP:conf/sigmod/HayMMCZ16,DBLP:conf/icde/YaroslavtsevCPS13}, the query function can be defined as the identity function as defined in Definition \ref{definition-permutation-function} where the codomain $\mathcal R$ of $f$ is the same as its domain $\mathcal D$. There will be different ways to set the metric $d$, in which the simplest way is to set $d=\bar d$, i.e., the metric induced by the $\ell_1$-norm. The synthetic dataset generation problem \cite{DBLP:conf/tcc/GuptaRU12,DBLP:conf/focs/HardtR10,DBLP:conf/stoc/RothR10,DBLP:conf/nips/HardtLM12} is a special case of the data publishing problem where the metric $d$ is the induced metric of the $\ell_2$-norm of the linear queries' output vector.
\end{example}

\begin{example}[principal component analysis \cite{DBLP:conf/stoc/HardtR12,DBLP:conf/stoc/DworkTT014,DBLP:conf/soda/KapralovT13,DBLP:conf/kdd/McSherryM09,DBLP:conf/stoc/HardtR13}]
For the principal component analysis problem, each record is a real-valued vector and a dataset $x$ is an $n\times m$ real-valued matrix. The value $f(x)$ is the principal component analysis matrix of $x$, i.e., a $k\times m$ real-valued matrix. Then $\mathcal R$ is a set of $k\times m$ real-valued matrices. The metric $d(f(x),r)$ can be set as the spectral norm or the frobenius norm of the matrix $r-f(x)$.
\end{example}

\begin{example}[linear classifier \cite{DBLP:journals/jmlr/ChaudhuriMS11,DBLP:conf/nips/ChaudhuriM08,DBLP:journals/pvldb/ZhangZXYW12}]
For the linear classifier problem, the value $f(x)$ can be set as the classifier of the dataset $x$ (if there are several candidates, just choose one randomly), i.e., a $k$ dimensional real-valued vector, which is the output of a classifier algorithm, such as the logistic regression algorithm. Then $\mathcal R$ is a set of $k$ dimensional real-valued vectors and $d(r,f(x))=\|r-f(x)\|_p$.
\end{example}

\begin{example}[functional output \cite{DBLP:journals/jmlr/HallRW13,DBLP:journals/jmlr/ChaudhuriMS11,DBLP:journals/corr/abs-0911-5708}]
As noted in \cite{DBLP:journals/jmlr/HallRW13}, there are many applications, in which the outputs are functions, such as the density functions. In this case, the value of $f(x)$ would be a function and the codomain $\mathcal R$ would be a set of functions. Let $(\mathcal F,\|\cdot\|)$ be a normed space \cite{Erwin-Kreyszig1978} and let $\mathcal R\subseteq \mathcal F$, then $d(r,f(x))=\|r-f(x)\|$.
\end{example}

\section{Representation Theory to Mechanisms} \label{section-representation-theory}

The main obstacle to solve the optimal mechanism problems (\ref{equation-23}) and (\ref{equation-25}) is that we are almost unknown to the set $\mathbb B$, the universe of the $\epsilon$-differential privacy mechanisms of the query function $f$. The set $\mathbb B$ is complex since, for each mechanism $\{p^x(r),x\in \mathcal D\} \in \mathbb B$, the probability distributions $p^x(r),x\in \mathcal D$ are strongly correlated to each other. In this section, we first study the correlations among the probability distributions $p^x(r),x\in \mathcal D$, and then study the representation of mechanisms and the structure of the set $\mathbb B$. Our study is motivated by some approaches to study the representation of elements in Hilbert space, where each element in a Hilbert space can be uniquely represented by the Fourier coefficients through a orthonormal basis \cite[Chapter 3]{Erwin-Kreyszig1978}.
% operators and the structure of the operator space in functional analysis \cite[Chapter 3]{Erwin-Kreyszig1978}. 
 Note that we don't mean that the set $\mathbb B$ has similar structure with Hilbert spaces. In fact, they are very different.

Let $f:\mathcal D \rightarrow \mathcal R$ be any one query function as defined in Section \ref{section-DP-model}. For the clarity of presentation, in the following parts of this paper, we assume that both $\mathcal D$ and $\mathcal R$ are discrete. 
In this setting, for any $A\subseteq \mathcal R$, we set $\int_A dr = |A|$. Other cases can be treated similarly.

\subsection{Discretization of Mechanism}

The discretization of mechanism is our first step to characterize the correlations among $p^x(r),x\in \mathcal D$. Our idea is to substitute the study of the discretized mechanism for the study of its original mechanisms.

\begin{definition}[Discretization of Mechanism]  \label{definition-discretization}
For the mechanism $\{p^x(r):x\in \mathcal D\}$, let $M=\max_{r\in \mathcal R,x\in \mathcal D} p^x(r)$.
For each $x\in \mathcal D$, set
\begin{align}  \label{equation-8}
\mathcal R_i^x =\left\{r \in \mathcal R: \exp(-(i+1)\epsilon) < \frac{p^x(r)}{M} \le \exp(-i\epsilon) \right\}, i\in \bar{\mathbb N}.
\end{align}
Then the mechanism $\{q^x(r):x\in \mathcal D\}$ is said to be the discretization (or discretized mechanism) of $\{p^x(r):x\in \mathcal D\}$ if
\begin{align}
q^x(r)=\frac{\exp(-i\epsilon)}{\alpha^x} , \mbox{  for all } r\in \mathcal R_i^x,
\end{align}
where $\alpha^x=\sum_{i=0}^{\infty}\exp(-i\epsilon) \left|\mathcal R_i^x \right|$ is the normalizer.
\end{definition}

We now show that the discretized mechanism $\{q^x(r):x\in \mathcal D\}$ has similar privacy level and utility level with its original mechanism $\{p^x(r):x\in \mathcal D\}$.

\begin{theorem}  \label{theorem-approximation}
Let $\{p^x(r):x\in \mathcal D\}$ and $\{q^x(r):x\in \mathcal D\}$ be as shown in Definition \ref{definition-discretization}. If $\{p^x(r):x\in \mathcal D\}$ is $\epsilon$-differentially private, then $\{q^x(r):x\in \mathcal D\}$ at least satisfies $2\epsilon$-differential privacy. 
Furthermore, letting $\mathbb E_p d(f(x),r), \mathbb E_q d(f(x),r)$ denote the utility of $\{p^x(r):x\in \mathcal D\}$ and $\{q^x(r):x\in \mathcal D\}$ at the dataset $x$, respectively, as shown in (\ref{equation-16}), then we have
\begin{align}\exp(-\epsilon)\mathbb E_q d(f(x),r)\le \mathbb E_p d(f(x),r)\le \exp(\epsilon)\mathbb E_q d(f(x),r).
\end{align}
%\alert{Furthermore,  their accuracies compared to the one of constructed are at most $\exp(\epsilon)$ times.}
\end{theorem}

\begin{proof}
We first prove the claim about privacy. 
Let $\{p^x(r):x\in \mathcal D\}$ be $\epsilon$-differentially private and let $x,x'\in \mathcal D$ be two neighbors.
For any $r\in \mathcal R$, assume $q^x(r)=\frac{\exp(-i\epsilon)}{\alpha^x}$ and $q^{x'}(r)=\frac{\exp(-j\epsilon)}{\alpha^{x'}}$.
%Set $p^x(r_0)=M^x$. Then, for any two neighbors $x,x'\in \mathcal D$, there is $\frac{M^x}{M^{x'}} \le \frac{p^{x}(r_0)}{p^{x'}(r_0)}\le e^{\epsilon}$ since $\{p^x(r):x\in \mathcal D\}$ satisfies $\epsilon$-differential privacy.
We have
\begin{align}
%\frac{q^x(r)}{q^{x'}(r)}=\frac{\exp(-i\epsilon)}{\exp(-j\epsilon)}=
\frac{p^x(r)}{M}/\frac{p^{x'}(r)}{M} = \frac{p^x(r)}{p^{x'}(r)}  \le \exp(\epsilon),
\end{align}
which ensures $|i-j|\le 1$ by the equality (\ref{equation-8}).
Then
\begin{align}
\frac{q^x(r)}{q^{x'}(r)} =\frac{\exp(-i\epsilon)}{\exp(-j\epsilon)} / \frac{\alpha^{x'}}{\alpha^x}<  \exp(2\epsilon)
\end{align}
for any $r\in \mathcal R$, which ensures the $2\epsilon$-differential privacy of $\{q^x(r):x\in \mathcal D\}$.

Now, we prove the claim about utility. By (\ref{equation-16}), we have 
\begin{align}
\mathbb E_p d(f(x),r) =& \sum_{r\in \mathcal R}p^x(r)d(f(x),r)      \\
=& \sum_{i=0}^{\infty} \sum_{r\in \mathcal R_i^x}p^x(r) d(f(x),r)                          \\
=& \sum_{i=0}^{\infty} \sum_{r\in \mathcal R_i^x}\frac{p^x(r)}{\sum_{r\in\mathcal R}p^x(r)} d(f(x),r)  \\
\le& \sum_{i=0}^{\infty} \sum_{r\in \mathcal R_i^x}\frac{\exp(-i\epsilon)M}{ \sum_{i=0}^{\infty}\sum_{r\in \mathcal R_i^x}\exp(-(i+1)\epsilon)M } d(f(x),r)  \\
=& \exp(\epsilon)\sum_{i=0}^{\infty} \sum_{r\in \mathcal R_i^x}\frac{\exp(-i\epsilon)}{ \sum_{i=0}^{\infty}\left| \mathcal R_i^x\right|\exp(-i\epsilon) } d(f(x),r)  \\
=& \exp(\epsilon)\mathbb E_q d(f(x),r).
\end{align}

Similarly, we have
\begin{align}
\mathbb E_p d(f(x),r)
=& \sum_{i=0}^{\infty} \sum_{r\in \mathcal R_i^x}\frac{p^x(r)}{\sum_{r\in\mathcal R}p^x(r)} d(f(x),r)  \\
\ge& \sum_{i=0}^{\infty} \sum_{r\in \mathcal R_i^x}\frac{\exp(-(i+1)\epsilon)M}{ \sum_{i=0}^{\infty}\sum_{r\in \mathcal R_i^x}\exp(-i\epsilon)M } d(f(x),r)  \\
=& \exp(-\epsilon)\sum_{i=0}^{\infty} \sum_{r\in \mathcal R_i^x}\frac{\exp(-i\epsilon)}{ \sum_{i=0}^{\infty}\left| \mathcal R_i^x\right|\exp(-i\epsilon) } d(f(x),r)  \\
=& \exp(-\epsilon)\mathbb E_q d(f(x),r).
\end{align}

The claims are proved.
\qed
\end{proof}

Theorem \ref{theorem-approximation} implies that we can approximately substitute the study of the discretized mechanism for the study of its original mechanism with less deviation from both utility and privacy.

\begin{definition}[Equivalence of Mechanisms \uppercase\expandafter{\romannumeral1}]  \label{definition-equivalence-mechanisms}
Let $\{p^x(r):x\in \mathcal D\}$ and $\mathcal R^x_i$ be as shown in Definition \ref{definition-discretization}.
Then $\{\mathcal R^x_i: i\in \bar{\mathbb N}, x\in \mathcal D\}$ is called the set sequences of the mechanism $\{p^x(r):x\in \mathcal D\}$, and the set $\mathcal R^x_i$ is called the $i$\emph{th} layer of the probability distribution $p^x(r)$.
\emph{Two mechanisms are said to be equivalent if they share the same set sequences.}
Furthermore, if a point $r\in \mathcal R$ is in the $i$\emph{th} layer of the probability distribution $p^x(r)$, we denote ${L}_p^{x}(r)=i$.
\end{definition}

Note that the set sequences $\{\mathcal R^x_i: i\in \bar{\mathbb N}, x\in \mathcal D\}$ contains all the information to construct the mechanism $\{q^x(r):x\in \mathcal D\}$. Therefore, we can use the former to denote the later or to denote the equivalent mechanism $\{p^x(r):x\in \mathcal D\}$ when there is no ambiguity.

\subsection{Representation of Mechanisms}

We now discuss the properties of the set sequences, which accurately capture the correlations among $p^x(r),x\in \mathcal D$.

\begin{theorem}  \label{lemma-4}
Let $\{p^x(r):x\in \mathcal D\}$ and $\{\mathcal R^x_i: i\in \bar{\mathbb N}, x\in \mathcal D\}$ be as shown in Definition \ref{definition-discretization}. Set $\mathcal A_{i}^x= \cup_{j=0}^i\mathcal R_{j}^x$ for each $x\in \mathcal D$ and each $i \in \bar{\mathbb N}$.
If $\{p^x(r):x\in \mathcal D\}$ satisfies $\epsilon$-differential privacy, then
\begin{enumerate}
\item for any two neighbors $x,x'\in \mathcal D$, there is $\mathcal R^x_i\cap \mathcal R^{x'}_j \ne \emptyset$ only if $|i-j|\le 1$;
\item there is
\begin{align}  \label{equation-9}
\mathcal R_i^x\supseteq\cup_{x'\in \mathcal D:0<\bar d(x,x')\le 1}\mathcal R_{i-1}^{x'}-\mathcal A_{i-1}^x
\end{align}
for each $x\in \mathcal D$ and each $i \in \mathbb N$;
\item there is
\begin{align} \label{equation-12}
\mathcal A_{i}^x \supseteq \cup_{y\in \mathcal D:\bar d(x,y) \le i}\mathcal R_{0}^y
\end{align}
for each $x\in \mathcal D$ and each $i \in\bar{ \mathbb N}$;
\item there is
\begin{align}  \label{equation-11}
\mathcal R_i^x\supseteq \cup_{y\in \mathcal D: i-1 <\bar d(x,y) \le i}\mathcal R_{0}^y -\mathcal A_{i-1}^x
\end{align}
for each $x\in \mathcal D$ and each $i \in \mathbb N$.
\end{enumerate}
\end{theorem}

\begin{proof}
We prove the first claim by contradiction. Assume that there exists $r_0\in \mathcal R$ such that $r_0 \in\mathcal R^x_i\cap \mathcal R^{x'}_{i+2}$. Then $\frac{p^x(r_0)/M}{p^{x'}(r_0)/M} =\frac{p^x(r_0)}{p^{x'}(r_0)}> \exp(\epsilon)$, which is contrary to the $\epsilon$-differential privacy of $\{p^x(r):x\in \mathcal D\}$. Other cases can be treated similarly. Therefore, the first claim is correct.

Now we prove the second claim. By the first claim, there is $\mathcal A_1^x \supseteq \mathcal R_0^{x'}$ for each neighbor $x'$ of $x$, which ensures
\begin{align}
\mathcal R_1^x\supseteq\cup_{x'\in \mathcal D:0<\bar d(x,x')\le 1}\mathcal R_{0}^{x'}-\mathcal A_{0}^x.
\end{align}
%Now assume that for $i=k$, there are
%\begin{align}
%\mathcal R_k^x\supseteq\cup_{x'\in \mathcal D:0<\bar d(x,x')\le 1}\mathcal R_{k-1}^{x'}-\mathcal A_{k-1}^x, \mbox{ for all } x\in \mathcal D.
%\end{align}
When $i \ge 2$, by the first claim, there are $\left( \mathcal R_{i-1}^x\cup \mathcal R_{i}^x\cup \mathcal R_{i+1}^x \right)\supseteq \mathcal R_{i}^{x'}$ for each neighbor $x'$ of $x$. Therefore, we have
\begin{align}
\mathcal R_{i+1}^x\supseteq\cup_{x'\in \mathcal D:0<\bar d(x,x')\le 1}\mathcal R_{i}^{x'}-\mathcal A_{i}^x.
\end{align}
The second claim is proved.

We prove the third claim by induction. First, it is easy to verify that the equation (\ref{equation-12}) is correct when $i=0$. Second, assume the equation (\ref{equation-12}) holds for any $x\in \mathcal D$ when $i\le k$, where $k \in \bar{\mathbb N}$.
We have
\begin{align*}
\mathcal A_{k+1}^x  \supseteq_a \left(\cup_{x'\in \mathcal D: 0<\bar d(x,x')\le 1}\mathcal A_k^{x'}\right) &\supseteq_b \left(\cup_{x'\in \mathcal D:0<\bar d(x,x')\le 1} \left( \cup_{y\in \mathcal D:\bar d(x',y)\le k} \mathcal R_0^y \right)\right) \\
&\supseteq_c \left( \cup_{y\in \mathcal D:\bar d(x,y)\le k+1} \mathcal R_0^y \right),
\end{align*}
where $\supseteq_a$ is due to the first claim and the definition of $\mathcal A_{i}^x$, $\supseteq_b$ is due to the assumption, and $\supseteq_c$ is due to the triangular inequality property of the metric $\bar d$.
The third claim is proved.

Now we prove the forth claim. Since $\mathcal A_{i}^x= \cup_{j=0}^i\mathcal R_{j}^x=\mathcal R_{i}^x\cup \mathcal A_{i-1}^x$, we have
\begin{align}
\mathcal R_i^x \supseteq& \mathcal A_{i}^x-\mathcal A_{i-1}^x  \\
  \supseteq& \cup_{y\in \mathcal D: d(x,y) \le i}\mathcal R_{0}^y-\mathcal A_{i-1}^x
\supseteq \cup_{y\in \mathcal D: i-1 < d(x,y) \le i}\mathcal R_{0}^y -\mathcal A_{i-1}^x.
\end{align}
The forth claim is proved.

The proof is complete.
\qed
\end{proof}

Theorem \ref{lemma-4} shows some important properties of differentially private mechanisms, which characterize the correlations among $\mathcal M(x), x\in \mathcal D$ well. These properties, especially the fourth one, are the basis of this paper to analyze differential privacy.
We are more interested in one special case of (\ref{equation-11}), where, for each $x\in \mathcal D$ and each $i \in \mathbb N$, there is
\begin{align}  \label{equation-10}
\mathcal R_i^x= \cup_{y\in \mathcal D: i-1 < d(x,y) \le i}\mathcal R_{0}^y -\mathcal A_{i-1}^x,
\end{align}
%for each $x\in \mathcal D$ and each $i \in \mathbb N$, 
which gives the boundary condition required to satisfy differential privacy.
This special case indicates an interesting phenomenon: the corresponding mechanism $\{\mathcal R^x_i: i\in \bar{\mathbb N}, x\in \mathcal D\}$ is determined only by the initial values $\{\mathcal R^x_0:x\in\mathcal D\}$ through the construction rule (\ref{equation-10}).
This phenomenon motivates us to explore whether any one mechanism has the similar concise representation.

In order to achieve the aim, we need to rewrite the construction rule (\ref{equation-11}), which is shown as follows.
For the mechanism $\{\mathcal R^x_i: i\in \bar{\mathbb N}, x\in \mathcal D\}$, set
\begin{align} \label{equation-31}
\tilde{\mathcal R}_i^x = \mathcal R_{i}^x - \left(\cup_{y\in \mathcal D: i-1 < \bar d(x,y) \le i}\mathcal R_{0}^y-\mathcal A_{i-1}^x \right)
\end{align}
for each $i\in \mathbb N$ and each $x\in \mathcal D$.
Assume there exists a set $I^x=\{i_0^x,i_1^x,\ldots,i_{t^x}^x\} \subseteq \mathbb N$ ($t^x$ may be $\infty$) such that $\tilde{\mathcal R}_i^x \ne \emptyset$ for each $i \in I^x=\{i_0^x,i_1^x,\ldots,i_{t^x}^x\}$, and that $\tilde{\mathcal R}_i^x = \emptyset$ for any $i \notin I^x$.
Then, we have
%Let the set $\tilde{\mathcal R}_i^x \subseteq \mathcal R$ satisfy $\tilde{\mathcal R}_i^x \cap \left(\cup_{y\in \mathcal D: d(x,y) \le i}\mathcal R_{0}^y \right)=\emptyset$ for each $i \in I^x=\{i_0^x,\ldots,i_t^x\}$, where each $i_j^x \in \mathbb N$ and $i_0^x<\cdots<i_t^x$.
%Let $\tilde{\mathcal R}_i^x=\emptyset$ for all $i \notin I^x=\{i_0^x,\ldots,i_t^x\}$.
%For each $i \in \mathbb N$, we set
\begin{align}  \label{equation-7}
\mathcal R_{i}^x=\tilde{\mathcal R}_i^x \cup \left(\cup_{y\in \mathcal D: i-1 <\bar d(x,y) \le i}\mathcal R_{0}^y-\mathcal A_{i-1}^x \right)
\end{align}
for each $i\in \mathbb N$ and each $x\in \mathcal D$.
Denote $\mathcal R^x_{I^x}= \{\tilde{\mathcal R}^x_{i_0^x}, \ldots, \tilde{\mathcal R}^x_{i_{t^x}^x}\}$. Then the mechanism $\{\mathcal R^x_i: i\in \bar{\mathbb N}, x\in \mathcal D\}$ is uniquely determined by the initial values $\{\mathcal R_0^x, \mathcal R^x_{I^x}: x\in \mathcal D\}$ and the construction rule (\ref{equation-7}).
In this manner, the equation (\ref{equation-10}) can be considered as a special case of the equation (\ref{equation-7}) where $\mathcal R^x_{I^x}=\emptyset$ for all $x\in \mathcal D$. We then have the follow representation theorem of mechanisms.

\begin{theorem}[Representation of Mechanism]  \label{theorem-representation}
Let $\{p^x(r):x\in \mathcal D\}$ and $\{\mathcal R^x_i: i\in \bar{\mathbb N}, x\in \mathcal D\}$ be as shown in Definition \ref{definition-discretization}.
If $\{p^x(r):x\in \mathcal D\}$ satisfies $\epsilon$-differential privacy, then $\{\mathcal R^x_i: i\in \bar{\mathbb N}, x\in \mathcal D\}$ is uniquely determined by the initial values $\{\mathcal R_0^x, \mathcal R^x_{I^x}: x\in \mathcal D\}$ throngh the construction rule (\ref{equation-7}).
\end{theorem}

The equation (\ref{equation-7}) accurately captures the correlations among the probability distributions $p^x(r), x\in\mathcal D$.
One very important thing is that Theorem \ref{theorem-representation} shows a way to represent and study mechanisms universally. That is, in order to figure out differential privacy, we only need to study the principles of setting the initial values $\{\mathcal R_0^x, \mathcal R^x_{I^x}: x\in \mathcal D\}$. Then Definition \ref{definition-equivalence-mechanisms} can be rewritten as follow.
%, which is the main result to study differential privacy in this paper
\begin{definition}[Equivalence of Mechanisms \uppercase\expandafter{\romannumeral2}]  \label{definition-1}
Two $\epsilon$-differential privacy mechanisms are said to be equivalent if they share the same initial values $\{\mathcal R_0^x, \mathcal R^x_{I^x}: x\in \mathcal D\}$. % are the same. %\alert{Note that equivalent mechanisms have similar utility.}

Furthermore, for an $\epsilon$-differential privacy mechanism, if its $\mathcal R^x_{I^x}=\emptyset$ for all $x\in \mathcal D$, it is called a basic mechanism. Otherwise, it is called a general mechanism.
\end{definition}

Noticing that the initial values $\{\mathcal R_0^x, \mathcal R^x_{I^x}: x\in \mathcal D\}$ can uniquely determine the mechanism $\{\mathcal R^x_i: i\in \bar{\mathbb N}, x\in \mathcal D\}$, we will use these initial values to denote the mechanism in the following sections where necessary.

\begin{proposition} \label{proposition-3}
For the mechanism $\{\mathcal R^x_i: i\in \bar{\mathbb N}, x\in \mathcal D\}$, if it can be reconstructed by the initial values $\{\mathcal R_0^x, \mathcal R^x_{I^x}: x\in \mathcal D\}$ through the construction rule (\ref{equation-7}), then it at least satisfies $2\epsilon$-differential privacy.
\end{proposition}

For the basic mechanisms, we have the following corollary.

\begin{corollary}  \label{corollary-1}
Let $\{p^x(r):x\in \mathcal D\}$ and $\{\mathcal R^x_i: i\in \bar{\mathbb N}, x\in \mathcal D\}$ be as shown in Definition \ref{definition-discretization}. %Set $\mathcal A_{i}^x= \cup_{j=0}^i\mathcal R_{j}^x$ for each $x\in \mathcal D$ and each $i \in \bar{\mathbb N}$.
If $\{p^x(r):x\in \mathcal D\}$ is a basic mechanism and satisfies $\epsilon$-differential privacy, then
\begin{enumerate}
\item for any two neighbors $x,x'\in \mathcal D$, $\mathcal R^x_i\cap \mathcal R^{x'}_j \ne \emptyset$ only if $|i-j|\le 1$;
\item there is
\begin{align}  \label{equation-13}
\mathcal R_i^x=\cup_{x'\in \mathcal D:0<\bar d(x,x')\le 1}\mathcal R_{i-1}^{x'}-\mathcal A_{i-1}^x
\end{align}
for each $x\in \mathcal D$ and each $i \in \mathbb N$;
\item there is
\begin{align} \label{equation-14}
\mathcal A_{i}^x = \cup_{y\in \mathcal D: \bar d(x,y) \le i}\mathcal R_{0}^y
\end{align}
for each $x\in \mathcal D$ and each $i \in\bar{ \mathbb N}$;
\item there is
\begin{align}  \label{equation-15}
\mathcal R_i^x= \cup_{y\in \mathcal D: i-1 <\bar d(x,y) \le i}\mathcal R_{0}^y -\mathcal A_{i-1}^x
\end{align}
for each $x\in \mathcal D$ and each $i \in \mathbb N$.
\end{enumerate}
\end{corollary}

\subsection{Representation of Optimal Mechanism Problems} \label{subsection:representation-utilities}

For the function $f$, let
\begin{equation}
\begin{aligned}
\bar {\mathbb B}= & \resizebox{.95\hsize}{!}{$     \left\{\{\mathcal R^x_i: i\in \bar{\mathbb N}, x\in \mathcal D\}:  \{\mathcal R^x_i: i\in \bar{\mathbb N}, x\in \mathcal D\} \mbox{ is the discretization of a mechanism in } \mathbb B \right\} 
$}\\
:=_a& \resizebox{.95\hsize}{!}{$        \left\{\{\mathcal R_0^x, \mathcal R^x_{I^x}: x\in \mathcal D\}: \{\mathcal R_0^x, \mathcal R^x_{I^x}: x\in \mathcal D\} \mbox{ is the initial values of a mechanism in } \bar {\mathbb B}\right\}    $}
\end{aligned} \label{equation-29}
\end{equation} 
\noindent
denote the universe of the discretized mechanisms of the mechanisms in $\mathbb B$, where the equality $=_a$ is due to Theorem \ref{theorem-representation}.
Then, the set $\bar{\mathbb B}$  would be an approximation of the set $\mathbb B$ and shows a beautiful structure of $\mathbb B$.
By Theorem \ref{theorem-approximation} and Theorem  \ref{theorem-representation}, the Pareto optimal mechanism problem (\ref{equation-23})
can be approximated by the Pareto optimal mechanism problem 
\begin{align} \label{equation-27}
\min_{\{\mathcal R_0^x, \mathcal R^x_{I^x}: x\in \mathcal D\} \in \bar{\mathbb B}}\{ \bar P^x: x \in \mathcal D\},
\end{align}
where
\begin{align}  \label{equation-18}
\bar P^x=\mathbb E_q d(f(x),r) =\frac{\sum_{i=0}^{\infty} \exp(-i\epsilon)\sum_{r\in \mathcal R_i^x} d(f(x),r)}{ \sum_{i=0}^{\infty}\exp(-i\epsilon)\left| \mathcal R_i^x\right| }.
\end{align}
Furthermore, 
let the set 
\begin{equation}
\begin{aligned}
\mathbb C=& \left\{\{\mathcal R^x_i: i\in \bar{\mathbb N}, x\in \mathcal D\}: \{\mathcal R^x_i: i\in \bar{\mathbb N}, x\in \mathcal D\} \mbox{ satisfies (\ref{equation-7})}\right\} \\
:=& \resizebox{.95\hsize}{!}{$        \left\{\{\mathcal R_0^x, \mathcal R^x_{I^x}: x\in \mathcal D\}: \{\mathcal R_0^x, \mathcal R^x_{I^x}: x\in \mathcal D\} \mbox{ is the initial values of a set sequences in } \mathbb C\right\}    $}
\end{aligned} \label{equation-28}
\end{equation}
denote the universe of the set sequences $\{\mathcal R^x_i: i\in \bar{\mathbb N}, x\in \mathcal D\} $ satisfying (\ref{equation-7}), where
\begin{align} 
\tilde{\mathcal R}_i^x \cap \left(\cup_{y\in \mathcal D: i-1 <\bar d(x,y) \le i}\mathcal R_{0}^y-\mathcal A_{i-1}^x \right) =\emptyset.
\end{align}
Then, there is $\mathbb C \supseteq \bar{\mathbb B}$. Unfortunately, we are unknown whether there is $\mathbb C =\bar{\mathbb B}$. Nevertheless, $\mathbb C$ is a good approximation of $\bar{\mathbb B}$ by Proposition \ref{proposition-3} and therefore the optimization problem (\ref{equation-23}) can be approximated by the optimization problem
\begin{align} \label{equation-24}
\min_{\{\mathcal R_0^x, \mathcal R^x_{I^x}: x\in \mathcal D\} \in \mathbb C}\{ \bar P^x: x \in \mathcal D\}.
\end{align}
Clearly, the problem (\ref{equation-24}) is much more operational than the problem (\ref{equation-23}) since the set $\mathbb C$ is known.
%$\{\mathcal R_0^x, \mathcal R^x_{I^x}: x\in \mathcal D\}$ are free parameters, i.e., each parameter is independent to other parameters in $\{\mathcal R_0^x, \mathcal R^x_{I^x}: x\in \mathcal D\}$,  but $\{p^x(r):  x\in \mathcal D\}$ are not.

Similarly, the optimal mechanism problem (\ref{equation-25}) can be approximated by the optimization problem
\begin{align} \label{equation-30}
\min_{\{\mathcal R_0^x, \mathcal R^x_{I^x}: x\in \mathcal D\}\in \mathbb C}\bar P,
\end{align}
where 
\begin{align}  
\bar P=\mathbb E \bar P^x=\sum_{x\in \mathcal D} p(x)\frac{\sum_{i=0}^{\infty} \exp(-i\epsilon)\sum_{r\in \mathcal R_i^x} d(f(x),r)}{ \sum_{i=0}^{\infty}\exp(-i\epsilon)\left| \mathcal R_i^x\right| }.
\end{align}

\section{Analytic Construction of Mechanisms}\label{section-mechanism-design}

The above section shows that, for the query function $f: \mathcal D \rightarrow \mathcal R$, each mechanism and then the utilities are completely determined by the parameters $\{\mathcal R_0^x, \mathcal R^x_{I^x}: x\in \mathcal D\}$, and that the optimal mechanisms can be approximately evaluated through evaluating (\ref{equation-24}) and (\ref{equation-30}).
Now we discuss the changing rules of utilities when tuning these parameters within the set $\mathbb C$.
Another work of this section is to classify the mechanisms according to different settings of these parameters.

\subsection{The Basic Mechanisms}  \label{subsection:construction-basic-mechanisms}

In this section, we consider the basic mechanisms, i.e., the setting $\mathcal R^x_{I^x}=\emptyset$ for all $x\in \mathcal D$ by Definition \ref{definition-1}.
Notice that, by the construction of the set sequences $\{\mathcal R_i^x\}_{x,i}$ in Definition \ref{definition-discretization}, the set $\mathcal R_0^x$ contains those points with the highest outputting probabilities. Therefore, in order to obtain better utility, at least the point $f(x)$ should be included in that set, i.e., there should be $f(x) \in \mathcal R_0^x$ for each $x\in \mathcal D$. We are more interested in the case where $\{f(x) \}= \mathcal R_0^x$ for all $x\in \mathcal D$, which may be the simplest mechanism since the probability distribution $q^x(r)$ is completely determined only by a point $f(x)$.

\begin{definition}[Purest Mechanism]
If $\mathcal R_0^x=\{f(x) \}$ and $\mathcal R_{I^x}^x=\emptyset$ for all $x\in \mathcal D$, then the mechanism $\{\mathcal R_0^x, \mathcal R^x_{I^x}: x\in \mathcal D\}$ is called the purest mechanism for the query function $f:\mathcal D \rightarrow \mathcal R$.
\end{definition}

The purest mechanism has the following interesting property.

\begin{proposition}  \label{proposition-1}
For the purest mechanism $\{\mathcal R^x_i: i\in \bar{\mathbb N}, x\in \mathcal D\}$, there are
\begin{align}
\mathcal R^x_i=\{f(y): i-1 <\bar d(x,y)\le i\} - \{f(y): \bar d(x,y)\le i-1\}
\end{align}
for $ i\in \mathbb N$.
Especially, if $f$ is a strictly monotonic function, then there are  $\mathcal R^x_i=\{f(y): i-1 <\bar d(x,y)\le i\}$ for $ i\in \mathbb N$.
\end{proposition}

\begin{proof}
The first claim is an immediate corollary of Corollary \ref{corollary-1}. The second claim holds since there is
\begin{align}  \label{equation-17}
\{f(y): i-1 <\bar d(x,y)\le i\} \cap \{f(y): \bar d(x,y)\le i-1\} =\emptyset
\end{align}
if $f$ is strictly monotonic by Definition \ref{definition-monotonic-function}.
\qed
\end{proof}

Proposition \ref{proposition-1} shows a very interesting phenomenon: To the dataset $x$, the $i$th layer $\mathcal R_i^x$ of the probability distribution $q^x(r)$ is (approximately, i.e., if the equation (\ref{equation-17}) holds) the set of values of $f$ over $\mathcal N_i^x$, where $\mathcal N_i^x$ is just the set of datasets $y$ whose distances to $x$ satisfy $i-1 <\bar d(x,y)\le i$. This phenomenon is especially useful to understand the structure of differential privacy displayed in Theorem \ref{lemma-4}, Theorem  \ref{theorem-representation} and Corollary \ref{corollary-1}. It can be considered as a microcosm of the structure.

The following definition, the atomic mechanism, presents a generalization to the purest mechnism. 
%We now give a generalization of the purest mechanism: the atomic mechanism. 
It is atomic since it can't be split into more slim mechanisms but, on the other hand, can be used to generate other mechanisms.

\begin{definition}[Atomic Mechanism]
If each initial value set $\mathcal R_0^x$ only contains one point in $\mathcal R$ for all $x\in \mathcal D$, then the basic mechanism $\{\mathcal R_0^x:x\in \mathcal D\}$ is called an atomic mechanism.
\end{definition}

We now discuss how to derive other basic mechanisms from the purest mechanism or the atomic mechanisms. Note that the we only need to change the initial values $\{\mathcal R_0^x: x\in \mathcal D\}$ of the atomic mechanisms  or the purest mechanism.
A natural way is to set $\mathcal R_0^x$ to be a $\delta$-neighborhood $N_{\delta}^{f(x)}:= \{r\in \mathcal R: d(f(x),r)\le \delta\}$ of $f(x)$ but not be the singleton set $\{f(x)\}$, where $\delta$ is a small positive number. The reason of the above setting is obvious: Those points $r$ being near to $f(x)$ respect to metric $d$ should obtain higher outputting probabilities in order to obtain better utility by (\ref{equation-18}). 
%This is the second mechanism.

\begin{definition}[$\delta$-neighborhood mechanism]
If $\mathcal R_0^x=N_{\delta^x}^{f(x)}$ and $\mathcal R_{I^x}^x=\emptyset$ for all $x\in \mathcal D$, then the mechanism $\{\mathcal R^x_i: i\in \bar{\mathbb N}, x\in \mathcal D\}$ is called a $\delta$-neighborhood mechanism of the query function $f:\mathcal D \rightarrow \mathcal R$.
\end{definition}

Notice that each dataset $x$ may have different radius $\delta^x$ to its neighborhood $N_{\delta^x}^{f(x)}$.
One misconception is that the more larger of $\delta^x$ the more better utility of the $\delta$-neighborhood mechanism.
This is totally wrong since when one $x$'s $\delta^x$ becomes larger, the utility of other datasets, in general, becomes worse since those points in $\mathcal R_0^x$, in general, are all included in the $i$th layer of the dataset $y$ if $\bar d(x,y)=i$ by Corollary \ref{corollary-1}.
 Therefore, there should be a balance among the radii $\delta^x$'s of the datasets.

We now discuss the changing rules of the set sequences when one dataset's initial values changes. For example, for the purest mechanism $\{\mathcal R_0^x:x\in \mathcal D\}$ and the mechanism $\{\bar{\mathcal R}_0^x:x\in \mathcal D\}$ such that  
$\bar{\mathcal R}^{x_0}_0=\{f(x_0), r_0\}$ and $\bar{\mathcal R}_0^x={\mathcal R}_0^x=\{f(x)\}$ for all $x\in \mathcal D\setminus\{x_0\}$, we may want to known the relation between the two mechanisms' set sequences. We have the following lemma for the purest mechanism.

\begin{lemma}  \label{lemma-5}
For the query function $f$, let $x_0, y_0$ be two datasets such that $f(x_0)\ne r_0$ where $r_0=f(y_0)$.
Let $\{q^x(r):x\in\mathcal D\}$ be the purest mechanism of $f$ and let $\{\mathcal R^x_i: i\in \bar{\mathbb N}, x\in \mathcal D\}$ be its set sequences. Let $\{\bar q^x(r):x\in\mathcal D\}$ be one basic mechanism of $f$ and let $\{\bar{\mathcal R}^x_i: i\in \bar{\mathbb N}, x\in \mathcal D\}$ be its set sequences such that
$\bar{\mathcal R}^{x_0}_0=\{f(x_0), r_0\}$ and $\bar{\mathcal R}_0^x=\{f(x)\}$ for all $x\in \mathcal D\setminus\{x_0\}$.
%Let $\mathrm{No}_q^{x_0}(r_0)=i$, $\mathrm{No}_q^{x}(r_0)=\ell$, $\mathrm{No}_{\bar q}^x(r_0)=k$ and $\mathrm{No}_{\bar q}^x(f(x_0))=j$. Let $\bar d(x, x_0)=s$.
Let $y$ be any  one dataset such that $r\ne r_0$, where $r=f(y)$.
Then, for any dataset $x\in \mathcal D$, there are ${L}_{\bar q}^{x}(r)={L}_q^{x}(r)$ and ${L}_{\bar q}^{x}(r_0)=\min\left\{\lceil\bar d(x_0, x)\rceil, {L}_q^{x}(r_0)\right\}$. 
\end{lemma}

\begin{proof}
We first prove the equality $L_{\bar q}^{x}(r)=L_q^{x}(r)$. Recall that, by Corollary \ref{corollary-1},  there are
\begin{align}    \label{equation-21}
\mathcal R_i^x= \cup_{y\in \mathcal D: i-1 <\bar d(x,y) \le i}\mathcal R_{0}^y - \cup_{y\in \mathcal D: \bar d(x,y) \le i-1}\mathcal R_{0}^y
\end{align}
and 
\begin{align}  \label{equation-22}
\bar{\mathcal R}_i^x= \cup_{y\in \mathcal D: i-1 <\bar d(x,y) \le i}\bar{\mathcal R}_{0}^y -\cup_{y\in \mathcal D: \bar d(x,y) \le i-1}\bar{\mathcal R}_{0}^y
\end{align}
for each $x\in \mathcal D$ and each $i \in \mathbb N$.
%, where $\bar{\mathcal A}_{i}^x=\cup_{j=0}^i\bar{\mathcal R}_j^x$. 
Then, by combining the settings of $\mathcal R_0^x$ and $ \bar{\mathcal R}_0^x$, if $r\in \mathcal R_i^x$ and $r\ne r_0$, we have $r\in \bar{\mathcal R}_i^x$, which ensures the equality $L_{\bar q}^{x}(r)=L_q^{x}(r)$. 

Now we prove the equality $L_{\bar q}^{x}(r_0)=\min\left\{\lceil\bar d(x_0, x)\rceil,L_q^{x}(r_0)\right\}$. 
%Note that $\mathrm{No}_{\bar q}^{x}(r_0)=\min\{\lceil\bar d(x_0, x)\rceil,\mathrm{No}_q^{x}(r_0)\}$
Denote $i_0=\lceil\bar d(x_0, x)\rceil $, $j_0=\lceil\bar d(y_0, x)\rceil $ and $k_0=L_q^{x}(r_0)$.  
By (\ref{equation-22}), 
there are totally three possible layers of $\bar q^x(r)$ in which the point $r_0$ can be: the $i_0$th, the $j_0$th or the $k_0$th layer of $\bar q^x(r)$. 
Note that there is $k_0\le j_0$ by the equality (\ref{equation-21}), which implies that $r_0$ can't be in the $j_0$th layer. Then there leave two possible cases: $i_0 \le k_0$ and $i_0>k_0$.
First, assume $i_0 \le k_0$. Then the point $r_0$ must be in the $i_0$th layer of $\bar q^x(r)$ by the equality (\ref{equation-22}). Next, assume $i_0>k_0$. Then the point $r_0$ must be in the $k_0$th layer of $\bar q^x(r)$ by the equality (\ref{equation-22}).
In all, there is $L_{\bar q}^{x}(r_0)=\min\{\lceil\bar d(x_0, x)\rceil,L_q^{x}(r_0)\}$.

The proof is complete.
\qed

\end{proof}

Applying the proof techniques of Lemma \ref{lemma-5} to other basic mechanisms, we have the following theorem.

\begin{theorem}  \label{corollary-2}
For the query function $f$, let $x_0, y_0$ be two datasets such that $f(x_0)\ne r_0$ where $r_0=f(y_0)$.
Let $\{q^x(r):x\in\mathcal D\}, \{\bar q^x(r):x\in\mathcal D\}$ be two basic mechanisms of $f$ and let $\{\mathcal R^x_i: i\in \bar{\mathbb N}, x\in \mathcal D\}, \{\bar{\mathcal R}^x_i: i\in \bar{\mathbb N}, x\in \mathcal D\}$ be their set sequences, respectively.
Assume $\mathcal R^x_{0}=\bar{\mathcal R}^x_{0}$ for all $x\in \mathcal D\setminus\{x_0\}$.
Assume $\bar{\mathcal R}^{x_0}_{0}=\mathcal R^{x_0}_{0}\cup\{r_0\}$, where $r_0\notin \mathcal R^{x_0}_{0}$.
Let $y$ be one dataset such that $r\ne r_0$, where $r=f(y)$.
Then there are $L_{\bar q}^{x}(r)=L_q^{x}(r)$ and $L_{\bar q}^{x}(r_0)=\min\left\{\lceil\bar d(x_0, x)\rceil,L_q^{x}(r_0)\right\}$.
\end{theorem}

The results of Lemma \ref{lemma-5} and Theorem \ref{corollary-2} show that, if $r\ne r_0$, then the layer of $r$ in $\bar q^x(r)$ is equal to its layer in $q^x(r)$.
They also show that the layer of the point $r_0$ is always ascending in each probability distribution $\bar q^x(r)$ compared to its layer in $q^x(r)$. Now we may wonder  what is the effect of the layer ascending of $r_0$ to the utilities $\{P^x:x\in \mathcal D\}$. The following theorem gives the question an answer.

\begin{theorem}  \label{lemma-6}
Let $\{q^x(r):x\in\mathcal D\}$ and $\{\bar q^x(r):x\in\mathcal D\}$ be as shown in Theorem \ref{corollary-2}.
%two basic mechanisms. Let $x_0, y_0$ be two datasets such that $f(x_0)\ne f(y_0)$  and set $r_0=f(y_0)$. For $x\in \mathcal D$, set $\mathrm{No}_q^x(r)= j_0$.
%Let the initial values of $\{q^x(r):x\in\mathcal D\}$ be $\mathcal R_0^x=\{f(x)\}$ for all $x\in \mathcal D$. Let the initial values of $\{\bar q^x(r):x\in\mathcal D\}$ be $\bar{\mathcal R}^{x_0}_0=\{f(x_0), r_0\}$ and $\bar{\mathcal R}_0^x=\{f(x)\}$ for all $x\in \mathcal D\setminus\{x_0\}$.
Let $\mathbb E_q[d(r,f(x))]$, $\mathbb E_{\bar q}[d(r,f(x))]$ be the expected distortion of $\{q^x(r):x\in\mathcal D\}$ and $\{\bar q^x(r):x\in\mathcal D\}$ at $x$, respectively, as defined in (\ref{equation-18}).
Then, for any one dataset $x$, if
\begin{align}  \label{equation-20}
\frac{ \sum_{i=0}^{\infty}\exp(-i\epsilon )\sum_{r \in \mathcal R_i^{x}} d(f(x),r)-\exp(-L_q^{x}(r_0)\epsilon)d(f(x), r_0)}{\sum_{i=0}^{\infty}\exp(-i\epsilon )| \mathcal R_i^{x}|-\exp(-L_q^{x}(r_0)\epsilon)} \le \frac{d(f(x),r_0)}{1},
\end{align}
there is $\mathbb E_q[d(r,f(x))] \le \mathbb E_{\bar q}[d(r,f(x))]$. Otherwise, there is $\mathbb E_q[d(r,f(x))] \ge \mathbb E_{\bar q}[d(r,f(x))]$.
\end{theorem}

\begin{proof}
This is an immediate corollary of Lemma \ref{chap-optimal-mechanism:lemma-1} and Theorem \ref{corollary-2}.
\qed
\end{proof}

Theorem \ref{lemma-6} shows that it is possible to improve the utility of the basic mechanism at $x\in \mathcal D$ by migrating some point $r_0$ into the initial values set of $x_0 \in \mathcal D$ so long as the distance $d(f(x),r_0)$ is small enough, which is the theoretical foundation of the claim that the $\delta$-neighborhood mechanisms may have more better utility than the purest mechanisms. However, Theorem \ref{lemma-6} only shows the change rule of utility when one $x$'s initial value set changes, the change rule of utility when several or all $x$'s initial value sets change simultaneously is still unclear.
%\alert{even if one can achieve the case one datset by another dataset}. 
We will give some experimental evidence to the change rule in Section \ref{section-application}, whose theoretical result needs further exploring.

%\alert{Now, we discuss how to iteratively generate other basic mechanisms from the atomic mechanisms by changing the initial value set.}

\begin{corollary}  \label{corollary-3}
Each basic mechanism can be derived from an atomic mechanism in the way of one dataset by another dataset as in Theorem \ref{corollary-2}.
\end{corollary}

\begin{proof}
Let $\{\bar{\mathcal R}_0^x:x\in \mathcal D\}$ be the mechanism needing to be generated.
First, set the mechanism $\{q^x(r):x\in\mathcal D\}$ to be the atomic mechanism where each $\mathcal R_0^x \subseteq \bar{\mathcal R}_0^x$ for all $x\in \mathcal D$.
The claim can then be proved just by iteratively using the result of Theorem \ref{corollary-2}.
\qed
\end{proof}

\subsection{The General Mechanisms}

This section considers the general mechanisms where there exist some $x\in \mathcal D$ such that $\mathcal R^x_{I^x} \ne \emptyset$ as defined in Definition \ref{definition-1}. We first extend the results of Theorem \ref{corollary-2} to the general mechanisms.

\begin{corollary}  \label{lemma-7}
For the query function $f$, let $x_0, y_0$ be two datasets such that $f(x_0)\ne r_0$ where $r_0=f(y_0)$.
Let $\{q^x(r):x\in\mathcal D\}, \{\bar q^x(r):x\in\mathcal D\}$ be two mechanisms of $f$ and let $\{\mathcal R^x_i: i\in \bar{\mathbb N}, x\in \mathcal D\}, \{\bar{\mathcal R}^x_i: i\in \bar{\mathbb N}, x\in \mathcal D\}$ be their set sequences, respectively.
Assume $\mathcal R^x_{0}=\bar{\mathcal R}^x_{0}$ and $\mathcal R^x_{I^x}=\bar{\mathcal R}^x_{I^x}$ for all $x\in \mathcal D\setminus\{x_0\}$.
Assume $\mathcal R^{x_0}_{i}=\bar{\mathcal R}^{x_0}_{i}$ for all $i \in I^{x_0}\cup\{0\}\setminus\{i_0\}$ and assume $\bar{\mathcal R}^{x_0}_{i_0}=\mathcal R^{x_0}_{i_0}\cup\{r_0\}$, where $i_0 \in I^{x_0}\cup\{0\}$ and 
%$r_0\notin \cup_{i \in I^{x_0}\cup\{0\}} \mathcal R^{x_0}_{i}$.
$r_0\notin \mathcal R^{x_0}_{i_0}$.
Let $y$ be one dataset such that $r\ne r_0$, where $r=f(y)$.
Then there are $L_{\bar q}^{x}(r)=L_q^{x}(r)$ and $L_{\bar q}^{x}(r_0)=\min\left\{\lceil\bar d(x_0, x)\rceil+i_0,L_q^{x}(r_0)\right\}$.
\end{corollary}

\begin{proof}
This is an immediate corollary of the construction rule (\ref{equation-7}) and Theorem \ref{corollary-2}.
\qed
\end{proof}

Note that the results of Theorem \ref{lemma-6} are also suitable for the general mechanisms. Therefore, the construction of mechanisms of general mechanisms is similar with the construction of the basic mechanisms in Section \ref{subsection:construction-basic-mechanisms}.

%\alert{Derive general mechanism from a basic mechanism while preserving/satisfying the construction rule (\ref{equation-7})}

\begin{corollary}
Each general mechanism can be derived from a basic mechanism in the way of one dataset by another dataset as in Corollary \ref{lemma-7}.
\end{corollary}

\begin{proof}
The proof is similar with the proof of Corollary \ref{corollary-3}.
\qed
\end{proof}

\subsubsection{The Approximation Mechanisms}

We now discuss the approximation problem. Let $f$ and $g$ be two functions with the domain $\mathcal D$ and the codomain $\mathcal R$.
The approximation problem is to substitute the study of differential privacy problem of $f$ for the study of differential privacy problem of the function $g$. 
One reason to study the approximation problem is the need of low sensitivity Lipschitz functions \cite{DBLP:journals/pvldb/KarwaRSY11,DBLP:conf/tcc/KasiviswanathanNRS13,DBLP:conf/sigmod/ChenZ13,DBLP:conf/focs/RaskhodnikovaS16}.
%, which is used as a low sensitivity function to substitute the original query function.
For example, if $g$ is a  low sensitivity Lipschitz function and is an approximation of $f$, then the  global sensitivity-based mechanisms, such as Laplace mechanism, can be applied to $g$ to obtain differentially private approximation of $f$ \cite{DBLP:conf/focs/RaskhodnikovaS16,DBLP:conf/tcc/DixitJRT13}. 

Another reason to study the approximation problem is for the situation where the datasets have different occurring probabilities. In this situation, it is reasonable to let those datasets with high occurring probabilities have relatively better utilities than those without. In this setting, the utility of a mechanism is measured by the quantity $P$ as shown in (\ref{equation-19}).

The key feature of the approximation problem is that when the dataset is $x$, the query result is not $f(x)$ or the points close to $f(x)$ with respect to metric $d$, which is formalized as the following definition.

\begin{definition}[Approximation Mechanism]
For a dataset $x$ and a positive number $\delta^x$, set $N_{\delta}^{f(x)}= \{r\in \mathcal R: d(r,f(x)) \le \delta\}$. Let $\{q^x(r): x\in \mathcal D\}$ be a mechanism for the query function $f(x),x\in\mathcal D$.
If there exists a dataset $x\in \mathcal D$ such that $N_{\delta^x}^{f(x)} \not\subseteq \mathcal R^x_0$ for any positive number $\delta^x$, we say that $\{q^x(r): x\in \mathcal D\}$ is an approximation mechanism for $f$.
\end{definition}

The aim of the approximation problems can be explained as substituting the study of the approximation mechanisms of $f$ for the study of the non-approximation mechanisms of $g$. 

\begin{proposition} \label{proposition-2}
Let $\{{\mathcal R}_0^x, {\mathcal R}^x_{I^x}: x\in \mathcal D\}$ be one approximation mechanism of $f$. Assume that, for each $x\in \mathcal D$, there exist one $f(x')\in \mathcal R$ and one positive number $\delta^{x'}$ such that $N_{\delta^{x'}}^{f(x')} \in {\mathcal R}_0^x$. Set $g(x)=f(x')$ for $x\in \mathcal D$. Then $\{{\mathcal R}_0^x, {\mathcal R}^x_{I^x}: x\in \mathcal D\}$ is one non-approximation mechanism of $g$.

\end{proposition}

\begin{proof}

The proof is immediate and therefore is omitted. 
\qed
\end{proof}

\subsection{Instance Analysis} \label{section-application}

In this section, we analyze some known mechanisms and give some experiments.

\subsubsection{$K$-Norm Mechanism}

The $K$-norm mechanism \cite{DBLP:conf/stoc/HardtT10,DBLP:conf/stoc/BhaskaraDKT12} is one interesting mechanism.
Set $\mathcal D=\mathbb R^n, \mathcal R=\mathbb R^d$ and let $F:\mathcal D \rightarrow \mathcal R$ be a linear function. In \cite{DBLP:conf/stoc/HardtT10,DBLP:conf/stoc/BhaskaraDKT12}, the metric $\bar d$ is defined as $\bar d(x,y)=\|x-y\|_1$ for any two datasets $x,y\in \mathcal D$, and the metric $d$ is defined as $d(r,r')=\|r-r'\|_2$ for any two points $r,r'\in \mathcal R$.
The $K$-norm mechanism is defined as 
\begin{align}
p^x(r)=\frac{1}{\alpha}\exp(-\|F(x)-r\|_K\epsilon),  
\end{align}
where $\alpha$ is the normalizer, $K=FB_1^n=\mathcal A_1^x-F(x)$ with $B_1^n=\mathcal N^x-x$ being the $\ell_1$ unit ball. 

\begin{proposition}
The $K$-norm mechanism is a basic mechanism.
\end{proposition}

\begin{proof}
By Definition \ref{definition-discretization}, 
\begin{align}
\mathcal R_i^x =&\{r \in \mathcal R : i \le \|F(x)-r\|_K < i+1\},  \\
%=&  \{r \in \mathcal R : i-1< \|F(x-x')+F(x')-r\|_K \le i\}
=& \{F(y) \in \mathcal R : i \le \|F(x-y)\|_K < i+1\} \\
=& \{F(y) \in \mathcal R : i \le \bar d(x,y) < i+1\} 
\end{align}
for $i \in \bar{\mathbb N}$ and $x\in \mathcal D$. Then 
\begin{align}
\mathcal A_i^x=&\{F(y) \in \mathcal R :  \bar d(x,y) < i+1\} 
= \cup_{y\in\mathcal D:\bar d(x,y)\le i}\mathcal R_{0}^{y}.
\end{align}
Therefore, 
\begin{align}
\mathcal R_i^x =& \mathcal A_i^x-\mathcal A_{i-1}^x  
= \cup_{y\in\mathcal D: i-1< \bar d(x,y)\le i}\mathcal R_{0}^{y}-\mathcal A_{i-1}^x,
\end{align}
which implies that the $K$-norm mechanism is a basic mechanism.
\qed
\end{proof}

\subsubsection{Sensitivity-Based Mechanisms}

The sensitivity-based mechanisms \cite{DBLP:conf/tcc/DworkMNS06,DBLP:conf/stoc/NissimRS07,DBLP:conf/sigmod/ZhangCPSX15} are a kind of mechanisms which are both simple and efficient. They include the global sensitivity-based mechanism \cite{DBLP:conf/tcc/DworkMNS06}, the local-sensitivity based mechanism \cite{DBLP:conf/sigmod/ZhangCPSX15}, the smooth sensitivity-based mechanism \cite{DBLP:conf/stoc/NissimRS07}, and some variants of constructing low-sensitivity Lipschitz functions \cite{DBLP:journals/pvldb/KarwaRSY11,DBLP:conf/tcc/KasiviswanathanNRS13,DBLP:conf/sigmod/ChenZ13,DBLP:conf/focs/RaskhodnikovaS16}. 

Now, we analyze the global sensitivity mechanism.
Let $p^x(r)$ be the probability distribution of the equation $(4)$ in \cite{DBLP:conf/tcc/DworkMNS06}. Then
\begin{align}
p^x(r)=\frac{1}{\alpha^x}\exp \left(-\frac{d(r,f(x))}{\Delta f}\epsilon \right),
\end{align}
where $\alpha^x$ is the normalizer. For simplicity, let $\alpha^x=\alpha^y$ for any two datases $x,y\in \mathcal D$. In the following, we prove that, in general, the global sensitivity-based mechanisms are not basic mechanisms for the non-monotonic functions.

\begin{proposition} \label{proposition-10}
For the query function $f:\mathcal D\rightarrow \mathcal R$, assume there exist three datasets $x_0,y_0,z_0\in \mathcal D$ such that $f(z_0) \notin \cup_{y\in\mathcal D: \bar d(x_0, y) \le i_0} \mathcal R_0^y$,
$f(z_0) \in \{r\in \mathcal R: i_0 \Delta f \le d(f(x_0),r) < (i_0+1)\Delta f\}$ and $\bar d(z_0,x_0) >i_0$, where $i_0=\lceil \bar d(y_0,x_0)\rceil$.
Then the global sensitivity-based mechanism is not a basic mechanism. 
\end{proposition}

\begin{proof}
Note that 
\begin{align*}
\mathcal R^{x_0}_{i_0}=&\{r\in \mathcal R: i_0\Delta f\le d(f(x_0),r) < (i_0+1)\Delta f\} \\
=& \{r\in \mathcal R: i_0\Delta f\le d(f(x_0),r) < (i_0+1)\Delta f\}-\{r\in \mathcal R: d(f(x_0),r) < i_0\Delta f\}
\end{align*}
and $\mathcal A_{i_0}^{x_0}=\{r\in \mathcal R: d(f(x_0),r) < (i_0+1)\Delta f\}$
for all $i\in \bar{\mathbb N}$. Furthermore, 
\begin{align*}
 &\cup_{y\in \mathcal D: i_0-1<\bar d(x_0,y)\le i_0}\mathcal R_{0}^{y}-\mathcal A_{i_0-1}^{x_0}\\
 =& \cup_{y\in\mathcal D: i_0-1<\bar d(x_0,y)\le i_0}\{r\in \mathcal R: 0\le d(f(y),r) < \Delta f\}-\{r\in \mathcal R: d(f(x_0),r) < i_0\Delta f\}. 
% \subseteq& \mathcal R_i^x.
\end{align*}
Then, by the assumptions, there is
\begin{align}
f(z_0) \notin \cup_{y: i_0-1<\bar d(x_0,y)\le i_0}\mathcal R_{0}^{y}-\mathcal A_{i_0-1}^{x_0}
\end{align}
but $f(z_0) \in \mathcal R_{i_0}^{x_0}$. Therefore, the global sensitivity-based mechanism is not a basic mechanism for $f$.
\qed
\end{proof}

Note that the assumptions about $f$ in Proposition \ref{proposition-10}, in general, are holding when $f$ is a non-monotonic function. Therefore, the global sensitivity-based mechanism is not a basic mechanism for   non-monotonic functions, in general.
Similarly, other sensitivity-based mechanisms, in general, are not basic mechanisms for non-monotonic query functions. 
The above analysis also shows that the sensitivity-based mechanisms are far less optimal respect to tradeoffs between utility and privacy for non-monotonic query functions since, in general, the set $\tilde {\mathcal R}_i^x$ is not empty set even  when $i$ is very large which will result in poor utility by Theorem \ref{lemma-6}. 

\subsubsection{Experiments to Subgraph Counting Function} 

The subgraph counting is one important problem in differential privacy \cite{DBLP:conf/sigmod/ChenZ13,DBLP:conf/tcc/KasiviswanathanNRS13,DBLP:journals/pvldb/KarwaRSY11,DBLP:conf/sigmod/ZhangCPSX15}. We use the \emph{edge differential privacy} as in \cite{DBLP:conf/sigmod/ZhangCPSX15}. That is, two graphs are said to be neighbors if the difference of their edges is 1. 
The query function $f$ is to count the number of \emph{triangles} in a graph $x\in \mathcal D$, where $\mathcal D = \mathcal G_n$ denotes the set of all the graphs with the number of nodes equals $n$ such that any two graphs in $\mathcal G_n$ are not isomorphic. 
Note that $\mathcal R=\{f(x):x \in \mathcal D\}$, which is the set of possible number of triangles. 
The utilities of mechanisms are measured by $\{P^x: x\in \mathcal D\}$ defined in (\ref{equation-16}).

We compare the utilities of the purest mechanism or the $\delta$-neighborhood mechanism to the Ladder mechanisms in \cite[Algorithm 1]{DBLP:conf/sigmod/ZhangCPSX15}. For the fairness of comparison, we set the codomain of the counting function $f$ in \cite[Algorithm 1]{DBLP:conf/sigmod/ZhangCPSX15} be $\mathcal R$ as above instead of $\mathbb Z$. Furthermore, we substitute $2\epsilon$ for $\epsilon$ in
\cite[Algorithm 1]{DBLP:conf/sigmod/ZhangCPSX15} which ensures that the Ladder mechanism is $2\epsilon$-differentially private as ours.
We evaluate the rate $R^x = P^x/P^x_L$, where $P^x$ denotes the utility of our mechanism at $x$ and $P^x_L$ denotes the corresponding utility of the Ladder mechanism at $x$.

\begin{figure}[t]
\includegraphics[scale=0.7]{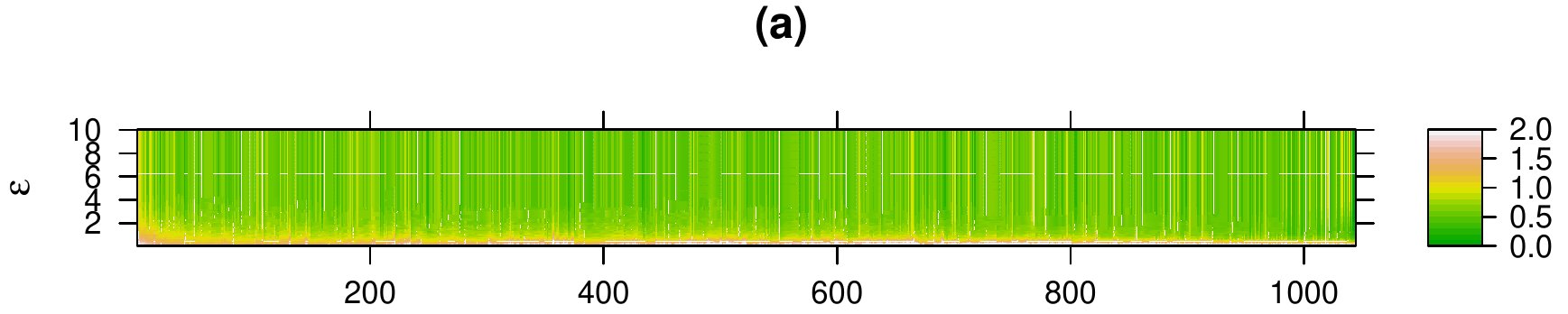}
\includegraphics[scale=0.7]{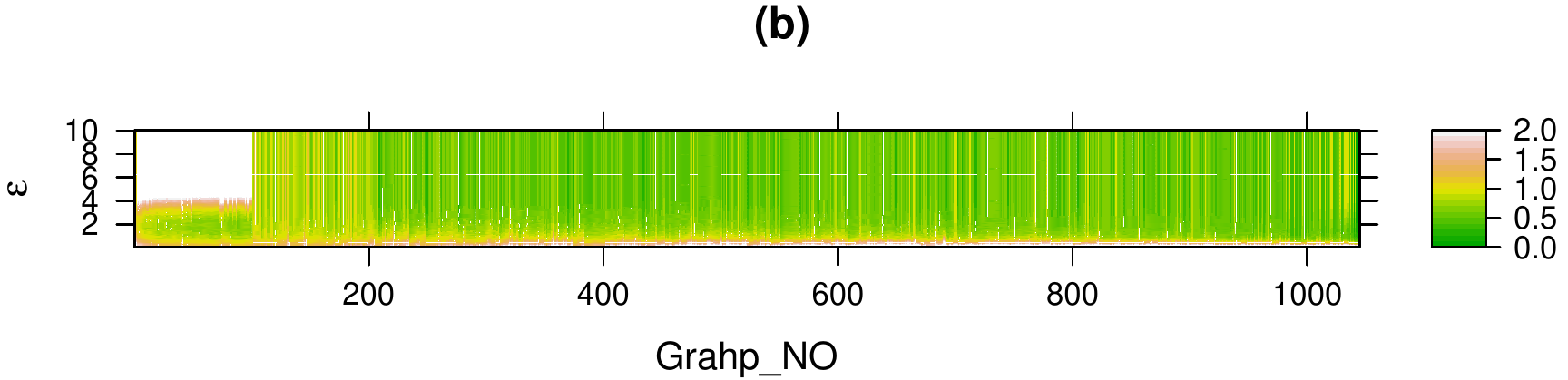}
\caption{The comparison of the purest mechanism and the $\delta$-neighborhood mechanism to the Ladder mechanism of their utilities for $\mathcal G_7$} \label{figure-ladder-ours-mean-G7}
\end{figure}
%\begin{figure}
%\includegraphics[scale=0.7]{3-delta}
%\caption{The comparison of ours and the ladder function of their mean values for $\mathcal G_7$}  \label{figure-ladder-ours-mean-G7-delta}
%\end{figure}
%trim={<left> <lower> <right> <upper>}
%\includegraphics[trim=1cm 10cm 0cm 8cm,width=1\textwidth]{2-delta.eps}

The details of the experiments are as follows. We set $\mathcal D$ be $\mathcal G_7$, where $|\mathcal D|=1044$ and $|\mathcal R|=28$. The results are shown in Fig. \ref{figure-ladder-ours-mean-G7} where the point $i$ in the $x$-axis denotes a graph $x$, the point $y$ in $y$-axis denotes the value $\epsilon$ and the value at the coordinate $(i,y)$ is the value $R^i$ when the input graph is $i$ and $\epsilon=y$.

The upper figure in Fig. \ref{figure-ladder-ours-mean-G7} shows the results when comparing the Ladder mechanism to the purest mechanism. The below figure in Fig. \ref{figure-ladder-ours-mean-G7} shows the results when comparing the Ladder mechanism to the $\delta$-neighborhood mechanism, where $\delta=1$ for all the graphs in $\{i: i \in \{1,2,\ldots,100\}\}$, and $\delta=0$ for other graphs in $\mathcal D$. From Fig. \ref{figure-ladder-ours-mean-G7} we can see the purest mechanism is better than the Ladder mechanism for most graphs when $\epsilon\ge 0.5$. However, the $\delta$-neighborhood mechanism  is worse than the Ladder mechanism for those graphs of $\delta=1$ and of most $\epsilon$. We reason that this is due to the distance $d(f(i),y)\ge 1$ for most $y\in \mathcal R$, which may result in that the inequality (\ref{equation-20}) does not hold.

\subsubsection{Experiments to Linear Function}

The linear query function (Definition \ref{definition-linear-function}) is a kind of well studied query functions in differential privacy \cite{DBLP:conf/stoc/HardtT10,DBLP:conf/stoc/NikolovTZ13,DBLP:journals/tit/GengV16,DBLP:journals/tit/GengV16a,DBLP:journals/vldb/LiMHMR15}. Instead of treating batch linear queries, we treat a linear query. %Batch linear queries can be treated by the combination of the method of this section and the functions compression method as discussed in Section \ref{subsection-related-work}.

We consider a special kind of the linear queries: the sum query. For the sum query, one dataset can be denoted as its histogram $x \in \mathbb R^N$, with $x_i$ denoting the number of elements in $x$ of type $i$ \cite{DBLP:conf/stoc/NikolovTZ13,DBLP:journals/fttcs/DworkR14}. As discussed in Section \ref{subsection-theory-function} we use the neighboring set $\mathcal V_f$ of $f$ to denote the linear function $f$.

The details of the experiments are as follows. We consider four linear functions (over four differnt dataset universes) respectively. They are $\mathcal V_1= [0,1] \cup [1000,1001]$, $\mathcal V_2= [0,100] \cup [1000,1001]$, $\mathcal V_3= [0,500] \cup [1000,1001]$, $\mathcal V_4= [0,1001]$, where $[a,b]$ denotes the corresponding interval in $\mathbb R$. Note that the first three sets are all concave set, which are different from the condition of the standard $K$-norm mechanism \cite{DBLP:conf/stoc/HardtT10} whose neighboring set $K$ is convex. 
We compare the purest mechanism and the $\delta$-neighborhood mechanism to the Staircase mechanism \cite{DBLP:journals/tit/GengV16} for these queries. 
%The corresponding algorithms to evaluate the sequences $\{\mathcal I^x_i\}_i$ and $\{\mathcal R^x_i\}_i$ are similar with Algorithm \ref{algorithm-subgraph-counting} and are omitted. 
Note that, since a linear query is symmetric for different datasets, the utilities $P^x, P^y$ are the same for any two $x,y\in \mathcal D$. 
Therefore, we only need to evaluate the set sequences and $P^x$ for only one dataset. Note that the same value $\delta$ is assigned to different datasets due to the above symmetric property. Before giving the detailed experiments, we first present some theoretical results about linear queries.

\begin{corollary} \label{corollary-linear-function-mechanisms}
Let the mechanism $\{p^x(r): x\in \mathcal D\}$ be $\epsilon$-differentially private and let its discretized mechanism $\{q^x(r): x\in \mathcal D\}$ be either a purest mechanism or a $\delta$-neighborhood mechanism. Then $\{q^x(r): x\in \mathcal D\}$ is $\epsilon$-differentially private.
%The purest mechanism and the $\delta$-neighborhood mechanism are both $\epsilon$-differentially private for the linear queries.
\end{corollary}

The above corollary is due to the symmetric property presented above which leads to $\alpha^x =\alpha^y$  for any two datasets $x,y$. %\emph{Note that, for two neighbors $x,x'$, the mechanisms in Corollary \ref{corollary-linear-function-mechanisms} reach the bound of the inequality (\ref{formula-inequality-dp}) at all the points in $\mathcal I^x_i\cap \mathcal I^{x'}_{i+1}$ or all the points in $\mathcal R^x_i\cap \mathcal R^{x'}_{i+1}$, $i \in \bar{\mathbb N}$}.

We now discuss the convergence of the set sequences.

\begin{definition}[The convergence of set sequence] \label{definition-sequence-convergence}
Let $\mathcal V_f$ be a linear query function over $\mathbb R$. The corresponding set sequence $\{\mathcal R_i: i \in \bar{\mathbb N}\}$ is said to be convergent if there exist $a_n$ and $n$ such that $\mathcal R_n=a_n\pm [0,\Delta f]$ and $\mathcal R_{n+1}=a_n\pm [\Delta f,2\Delta f]$, where $\Delta f$ is the global sensitivity of $\mathcal V_f$.
\end{definition}

\begin{proposition}  \label{proposition-convergence}
Assume $\mathcal V_f=[0,a]\cup[b,c] $ is a linear query function, where $0<a<b<c$. Then the sequence $\{\mathcal R_i:i\in \bar{\mathbb N}\}$ of $f$ is convergent. 
% when $i \ge \frac{\Delta f}{c-b}$. \alert{Need verification.}
\end{proposition}

\begin{proof}
Note that the interval $[b,c]$ will generate the interval $[ib,ic]$ in $\mathcal R_i^x$. Setting $i \ge \frac{\Delta f}{c-b}$, we have $[(i-1)\Delta f, i\Delta f]\subseteq [ib,ic] $. This implies that $[i\Delta f, (i+1)\Delta f]\subseteq \mathcal R_{i+1} $. Then it is convergent.
\end{proof}
%The corresponding $\{\mathcal R_i:i\in \bar{\mathbb N}\}$ sequence has the similar convergence property with $\{\mathcal I_i:i\in \mathbb N\}$ as in Proposition \ref{proposition-convergence}.

\begin{figure}[t]
\centering
\includegraphics[scale=0.33]{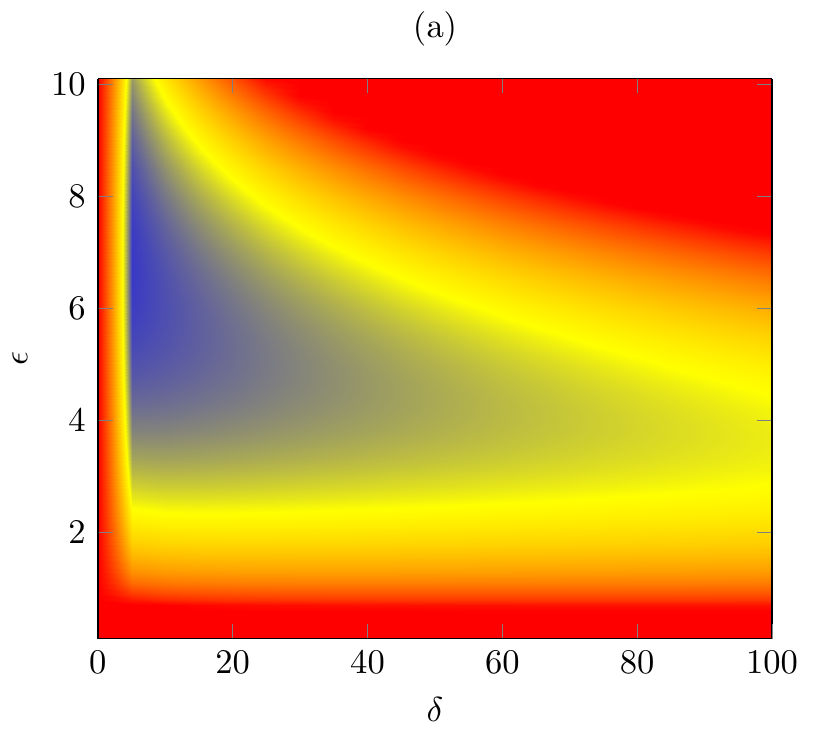}
\includegraphics[scale=0.33]{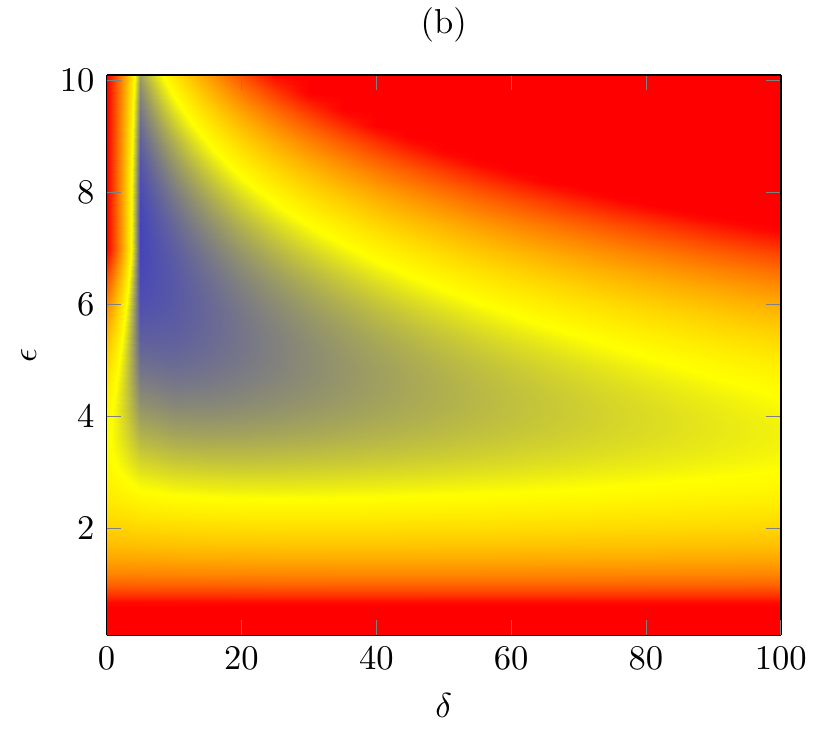}
\includegraphics[scale=0.33]{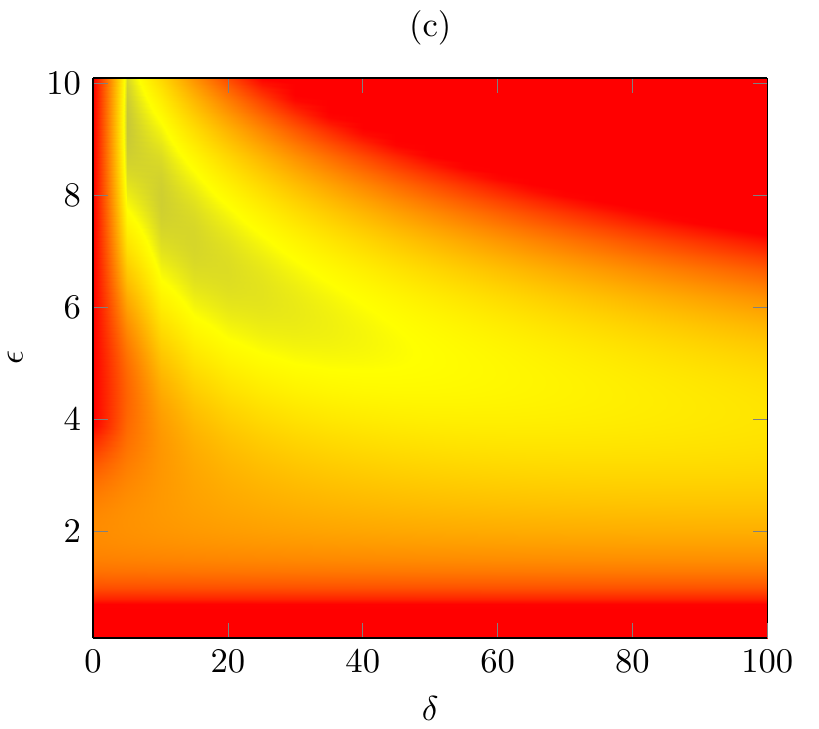}
\includegraphics[scale=0.33]{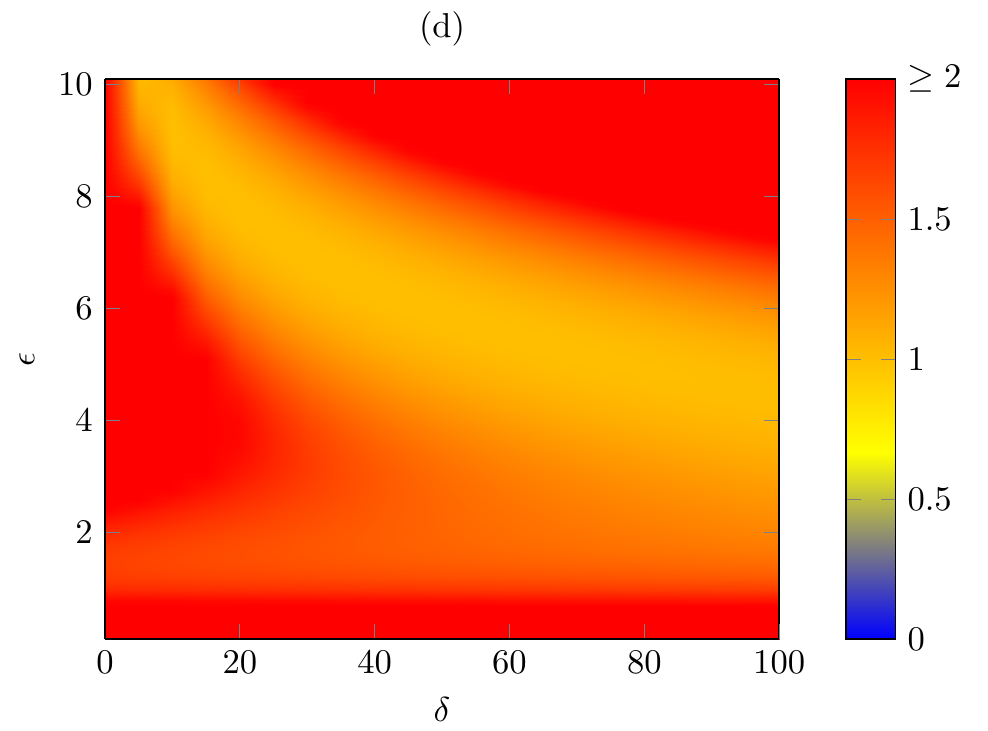}
\caption{The utilities' comparison of the purest mechanism and the $\delta$-neighborhood mechanism to the Staircase mechanism for the queries $\mathcal V_1,\mathcal V_2,\mathcal V_3,\mathcal V_4$ are shown in $(a), (b), (c), (d)$, respectively} \label{figure-example-2}
\end{figure}

The utility of the mechanism $\mathcal M$ is measured by $P^x$ defined in (\ref{equation-16}), in which $f(x)=0$. We compare the purest mechanism or the $\delta$-neighborhood mechanism to the Staircase mechanism in \cite[Algorithm 1]{DBLP:journals/tit/GengV16}. We evaluate the rate $R= P^x/P^x_S$, where $P^x$ denotes the utility of the purest mechanism or the $\delta$-neighborhood mechanism and $P^x_S$ denotes the utility of the Staircase mechanism. The results are shown in Fig. \ref{figure-example-2} where the $x$-axis denotes the values of $\delta$, the $y$-axis denotes the value of $\epsilon$ and the value at the coordinate $(\delta,\epsilon)$ is the value $R$ for the definite $\delta$ and $\epsilon$. In Fig. \ref{figure-example-2}, the results of the four queries $\mathcal V_1,\mathcal V_2,\mathcal V_3,\mathcal V_4$ are shown in $(a),(b),(c),(d)$, respectively.

We now analyze the results in Fig. \ref{figure-example-2}. The four queries $\mathcal V_1,\mathcal V_2,\mathcal V_3,\mathcal V_4$ have the same global (and local) sensitivity $\Delta f=1001$. However, the volumes of their neighboring set (Section \ref{subsection-theory-function}) are different. Explicitly, $\mbox{Vol}(\mathcal V_1)=2, \mbox{Vol}(\mathcal V_2)=101, \mbox{Vol}(\mathcal V_3)=501$ and $\mbox{Vol}(\mathcal V_4)=1001$. Fig. \ref{figure-example-2} shows some interesting phenomenon: The more larger of the value $\Delta f/\vol(\mathcal V_f)$, the more better of our mechanisms compared to the Staircase mechanism when $ \epsilon \ge 3 $ and $5 \le \delta \le 50$. Furthermore, Fig. \ref{figure-example-2} shows that it is possible that the $\delta$-neighborhood mechanism can have more better utilities than the purest mechanism at every datasets.

\section{Conclusion} \label{section-conclusion}

By capturing the correlations among the differential privacy outputs $\mathcal M(x), x\in \mathcal D$, the differential privacy mechanism $\mathcal M$ can be represented just by some parameters, by which the universe of $\epsilon$-differential privacy mechanisms of the query function $f$ is just a set of these parameters. These greatly simplify the construction of differential privacy mechanisms and then greatly simplify the optimal mechanism problems. More importantly, these results are universal to every query functions defined in Section \ref{section-DP-model}.

The results of this paper provide a way to universally discuss the optimal differentially private mechanisms defined in Section \ref{subsection:representation-utilities}, at least theoretically. Of course, we must acknowledge that the optimal mechanism problems are so complicated that it is too early to say that our approaches can give the optimal mechanism problems a relatively satisfied solution.  %Nevertheless, they provide a way to deal with the optimal mechanism problems. 
Clearly, it is necessary to give a detailed mathematical exploring of the optimal mechanisms for some simple functions, such as the permutation functions in Definition \ref{definition-permutation-function} and the linear functions in Definition \ref{definition-linear-function}, before discussing the optimal mechanisms of real world problems, like those examples in Section \ref{subsection:instance-interpretation}. It seems that some results about operators in functional analysis, especially about the linear operators, will inevitably be involved in the exploring. These should be urgent future works.

Furthermore, how to generalize the results of this paper to $(\epsilon, \delta)$-differential privacy is another interesting work. Moreover,  
this paper only focuses on the utility-privacy tradeoffs but seldom on the algorithm complexity. 
%We hope, in future, there will be works that focus on the tradeoff among all of utility, privacy and algorithm complexity.
We hope, in future, we can add the algorithm complexity consideration into our approaches.

\bibliographystyle{unsrt}

\end{document}